\documentclass[5p,times]{elsarticle}

\usepackage{amsmath,amssymb,amsfonts}

\usepackage{amsthm}

\newtheorem{theorem}{Theorem}
\newtheorem{proposition}{Proposition}

\newtheorem{corollary}{Corollary}

\usepackage{algorithmic}
\usepackage{url}
\usepackage{graphicx}
\usepackage{xcolor}
\usepackage{colortbl}
\usepackage{booktabs}
\definecolor{accessblue}{cmyk}{1 0.3 0 0.2}%
\definecolor{greycolor}{cmyk}{0,0,0,.8}

\usepackage{subcaption}
\usepackage{quantikz}
\usepackage{pgfplots}
\pgfplotsset{compat=1.18}
\usepgfplotslibrary{groupplots} 

\usepackage[acronym, nonumberlist, nogroupskip]{glossaries}

\def\BibTeX{{\rm B\kern-.05em{\sc i\kern-.025em b}\kern-.08em
    T\kern-.1667em\lower.7ex\hbox{E}\kern-.125emX}}

\newcommand{\qctrl}[1]{#1^{\hspace{0.1em}\bullet}}
\newcommand{\qtarg}[1]{#1^{\hspace{0.1em}\text{\scalebox{0.6}{$\square$}}}}

\makeglossaries

\newacronym{EPR}{EPR}{Einstein-Podolsky-Rosen}
\newacronym{GHZ}{GHZ}{Greenberger-Horne-Zeilinger}
\newacronym{QFT}{QFT}{Quantum Fourier Transform}
\newacronym{IQFT}{iQFT}{Inverse Quantum Fourier Transform}
\newacronym{QPE}{QPE}{Quantum Phase Estimation}

\usepackage{titlesec}

\titleformat{\paragraph}
{\normalfont\normalsize\itshape}{\theparagraph}{0.2em}{\normalfont\normalsize\itshape}
\titlespacing{\paragraph}
{0pt}{*1}{*1}[0.2pt]
\begin{document}

\begin{frontmatter}
    
\title{Communication-Efficient Distributed Inverse Quantum Fourier Transform}

\author[citius]{F. Javier Cardama}
\author[cesga]{Jorge Vázquez-Pérez}
\author[citius,dec]{Tomás F. Pena}
\author[cesga]{Andrés Gómez}

\affiliation[citius]{organization={Centro Singular de Investigación en Tecnoloxías Intelixentes (CiTIUS)},
            addressline={Universidade de Santiago de Compostela}, 
            city={Santiago de Compostela},
            postcode={15782}, 
            country={Spain}}
\affiliation[dec]{organization={Departamento de Electrónica e Computación},
            addressline={Universidade de Santiago de Compostela}, 
            city={Santiago de Compostela},
            postcode={15782}, 
            country={Spain}}
\affiliation[cesga]{organization={Galicia Supercomputing Center (CESGA)},
            addressline={Avda.\ de Vigo S/N}, 
            city={Santiago de Compostela},
            postcode={15705}, 
            country={Spain}}

\begin{abstract}
The scalability of quantum computing is currently limited by physical, technological, and architectural constraints that hinder the integration of a large number of qubits within a single quantum processor. Distributed quantum computing (DQC) has therefore emerged as a viable alternative, aiming to interconnect multiple smaller quantum processing units (QPUs) to jointly operate on a global quantum state. While this paradigm enables scalable architectures, it introduces significant communication overhead due to the cost of non-local quantum operations across distant nodes. In this work we propose a distributed formulation of the iQFT over a quantum network composed of $P$ nodes, each hosting $Q$ qubits, enabling the execution on a global register of size $n = P \cdot Q$. Furthermore, we introduce a communication-efficient variant based on a threshold-driven pruning strategy, referred to as a \emph{communication horizon}, which exploits the exponentially decreasing significance of controlled-phase rotations to safely omit remote gates with negligible impact. By reducing the number of inter-node quantum interactions, the proposed approach significantly lowers the quantum communication requirements of the distributed iQFT while operating within a strictly bounded approximation error. Crucially, we show that this approach fundamentally alters the scaling of the algorithm: the entanglement resource consumption per node saturates to a constant value, reducing the global communication complexity from quadratic $\mathcal{O}(P^2)$ to linear $\mathcal{O}(P)$. As the iQFT constitutes a critical building block in many quantum algorithms, the techniques presented in this paper directly contribute to improving the practicality and scalability of distributed quantum computation.
\end{abstract}

\end{frontmatter}

\section{Introduction}
\label{sec:introduction}
Distributed quantum computing (DQC) has emerged as a promising architectural approach to overcome the scalability limitations of current monolithic quantum processors~\cite{Caleffi2024DistributedSurvey, Barral2025ReviewComputing}, which are fundamentally constrained by qubit noise, decoherence, and the increasing complexity of control and calibration as the system size grows~\cite{Bluhm2019SemiconductorComputing,Preskill2018QuantumBeyond, Cacheiro2025QMIO:System}. Rather than relying on a single large quantum processor, DQC interconnects multiple smaller quantum processing units (QPUs) through quantum entanglement and classical communication channels, enabling the creation of a global quantum state spanning several physical devices~\cite{Wehner2018QuantumAhead, Ferrari2021CompilerComputing, Illiano2022QuantumSurvey}. By entangling qubits across a network of $P$ interconnected QPUs, each hosting $Q$ local qubits, a collective global register of size $n = P \cdot Q$ can be realized. This modular paradigm allows quantum systems to scale beyond the physical limits of individual processors~\cite{Escofet2025OnPerformance,Palesi2025AssessingSystems}.

Regarding quantum computing itself, it represents a fundamental shift in the computational paradigm, with several cornerstone quantum algorithms demonstrating its advantage, being the well known Shor’s integer factorization algorithm~\cite{Shor1994AlgorithmsFactoring} the most representative. Shor's algorithm leverages the Quantum Phase Estimation (QPE)~\cite{Jiang2021AnAlgorithms} algorithm to factor large integers in polynomial time, with the inverse Quantum Fourier Transform (iQFT) playing a central role in it, due to the QPE encoding its final results in the relative phases of quantum states expressed in the Fourier basis. The iQFT acts as a decoding operation that transforms this phase information back into the computational basis, enabling the extraction of classical outcomes through measurement. Consequently, the efficiency and scalability of the iQFT directly impact the performance of the QPE and, therefore, the Shor's algorithm---and all of those employing these oracle as part of their operations.

However, the standard implementation of the iQFT on an $n$-qubit register requires $\mathcal{O}(n^2)$ two-qubit controlled-phase gates and exhibits an all-to-all interaction pattern, as each qubit must, in principle, interact with every other qubit. This quadratic gate complexity and dense connectivity pose significant challenges for both monolithic and distributed quantum architectures. In distributed settings, the cost of non-local two-qubit gates is particularly high, as they require entanglement generation, quantum communication, and classical synchronization between remote nodes~\cite{Almudever2024FromSystems, Ferrari2023AComputing, Cardama2026NetQIR:Computing}.

Because of these demanding communication requirements, the distribution of the QFT and iQFT serves as a benchmark challenge in the field. Standard approaches rely on topology-aware qubit mapping~\cite{Sang2026}, remote-gate batching (such as the partitioning methods proposed by Kaur et al.~\cite{Kaur2025OptimizedQuantumCircuit}), and optimized circuit scheduling to efficiently route inter-node communications~\cite{Bahrani2024, Liu2025}. While these strategies minimize the practical execution time of network operations, they preserve the full logical structure of the exact circuit, meaning they do not reduce the fundamental asymptotic complexity of the communication overhead.

Recently, the theoretical limits of these exact implementations were formalized by Ebnenasir~\cite{Ebnenasir2025LowerBounds}. In this work, the author established the theoretical lower bounds for the communication costs of distributing the exact QFT on clique networks (quantum network with clique topology; i.e., complete graph, where any pair of quantum machines have a direct quantum link between them), rigorously proving that any exact execution necessitates an aggregate communication complexity of $\Omega(P^2)$ across $P$ processing nodes.

In this work, we address these challenges by proposing a communication-efficient distributed implementation of the iQFT. The main contributions of this paper are twofold:

\begin{itemize}
    \item\textbf{Distributed iQFT}: we present a systematic decomposition of the iQFT across a quantum network composed of $P$ nodes, each with $Q$ qubits, enabling the execution of the transform over a global register of size $n = P \cdot Q$.
    \item \textbf{Communication-efficient}: we introduce a communication-optimized variant of the distributed iQFT based on a well-known threshold-driven pruning strategy~\cite{Coppersmith2002ApproximateFourier}. By establishing a maximum communication horizon, this approach successfully reduces the global entanglement complexity of the network from quadratic $\mathcal{O}(P^2)$ to linear $\mathcal{O}(P)$ scaling.
\end{itemize}

The remainder of this paper is organized as follows. Section~\ref{sec:background} reviews the theoretical background of the iQFT and distributed quantum computing primitives. Section~\ref{sec:methods} introduces the proposed distributed iQFT formulation and the communication horizon optimization strategy. Section~\ref{sec:results_discussion} analyzes the impact of the proposed approach in terms of communication cost and accuracy. Finally, Section~\ref{sec:conclusions} conclude the paper and outlines directions for future research.

\section{Background}
\label{sec:background}

\subsection{Inverse Quantum Fourier Transform}
\label{sec:iqft_background}

The \gls{QFT} is a linear unitary transformation analogous to the classical discrete Fourier transform. For a quantum system of $n$ qubits, let $N = 2^n$ be the dimension of the Hilbert space. The QFT maps the orthonormal computational basis states $\{|x\rangle\}_{x=0}^{N-1}$ to a new set of orthonormal states, which we refer to as the \textit{Fourier basis} $\{|\widetilde{x}\rangle\}_{x=0}^{N-1}$. This operation is defined as:

\begin{equation}
    \text{QFT}|x\rangle = |\widetilde{x}\rangle = \frac{1}{\sqrt{N}} \sum_{y=0}^{N-1} e^{2\pi i \frac{xy}{N}} |y\rangle.
    \label{eq:qft_def}
\end{equation}

Each state $|\widetilde{x}\rangle$ in the Fourier basis is an equal superposition of all computational basis states, distinguished only by their relative phases. Consequently, the \gls{IQFT} serves as the decoding operation that maps a state from the Fourier basis back to the computational basis:

\begin{equation}
    \text{QFT}^{-1}|\widetilde{x}\rangle = |x\rangle.
    \label{eq:iqft_mapping}
\end{equation}

Mathematically, the iQFT is the adjoint of the QFT operator ($\text{QFT}^\dagger$). Its action on an arbitrary computational basis state $|y\rangle$ is given by introducing a negative sign in the complex exponential:

\begin{equation}
    \text{QFT}^{-1}|y\rangle = \frac{1}{\sqrt{N}} \sum_{x=0}^{N-1} e^{-2\pi i \frac{xy}{N}} |x\rangle.
    \label{eq:iqft_sum}
\end{equation}

This transformation is fundamental in quantum algorithms such as \gls{QPE} and Shor's algorithm, where the solution is encoded in the relative phases of the qubits (Fourier basis) and must be transformed back to the computational basis to be measured with high probability.

\subsubsection{Circuit Implementation}
The decomposition of the iQFT into a quantum circuit consists of a structured sequence of Hadamard ($H$) gates and two-qubit Controlled-Phase ($CP$) rotations, followed by a SWAP network to reverse the order of the qubits. The specific circuit implementation considered in this work is illustrated in Figure~\ref{fig:iqft_circuit}.

\begin{figure*}
    \centering
    \input{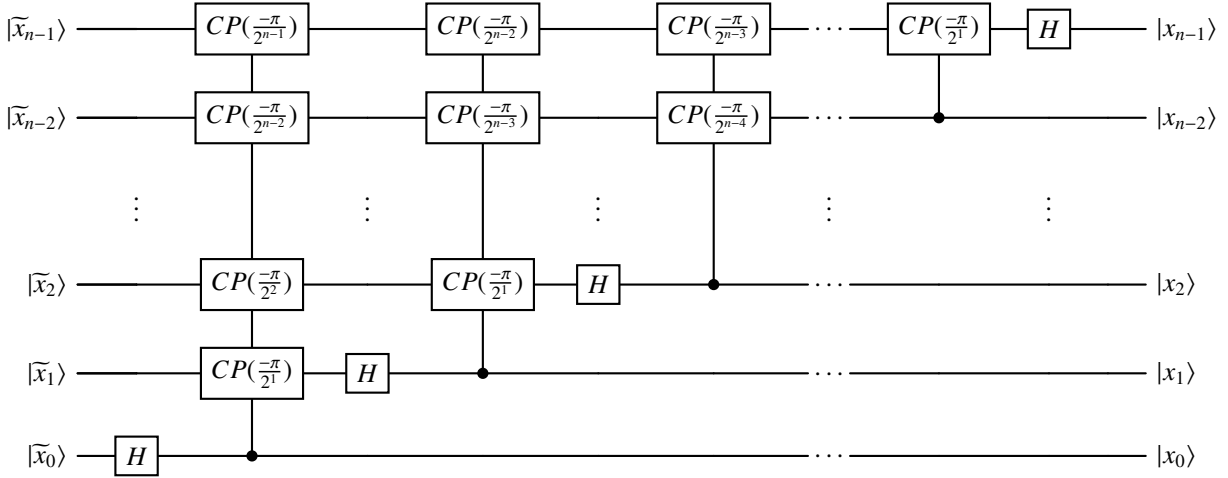}
    \caption{Circuit realization of the $n$-qubit Inverse Quantum Fourier Transform (iQFT). The circuit maps the input state from the Fourier basis $|\widetilde{x}\rangle$ to the computational basis $|x\rangle$ using a sequence of Hadamard gates ($H$) and controlled phases $CP(\theta_k)$ where $\theta_k = \pi/2^k$.}    \label{fig:iqft_circuit}
\end{figure*}

The rotation angle applied to a target qubit $\qtarg{i}$, controlled by qubit $\qctrl{i}$, is determined by the logical index difference $k$ between them, denoted as:

\begin{equation}\label{eq:theta_generic_iqft}
    \theta(\qctrl{i},\qtarg{i})=\theta_{\qtarg{i}-\qctrl{i}}= -\frac{\pi}{2^{\qtarg{i}-\qctrl{i}}}.
\end{equation}

$CP(\theta_{k})$, applies a diagonal phase shift depending on the index difference between the qubits ($k=\qtarg{i}-\qctrl{i}$). For the inverse transform, the rotation angle is negative:

\begin{equation}
    CP(\theta_k) = \text{diag}\left(1, 1, 1, e^{-i\frac{\pi}{2^{k}}}\right),
    \label{eq:cp_gate}
\end{equation}

\noindent where $k$ represents index difference between the control and target qubits in the register. The exact implementation requires $\mathcal{O}(n^2)$ such gates. However, as $k$ increases, the rotation angle approaches zero.

\begin{equation}
    \lim_{k\to\infty}\theta_k =0.
\end{equation}

This property is crucial for distributed implementations, as it allows for the truncation of small-angle rotations (approximate iQFT) to reduce inter-node communication, a concept we detail in subsequent sections.

To illustrate the specific connectivity pattern by these operations, Figure~\ref{fig:iqft_detail} describes the interactions between an arbitrary pair qubits, $\qctrl{i}$ as control and qubit $\qtarg{i}$ as target. The angle $\theta$ applied follows the Eq.~\ref{eq:theta_generic_iqft}:

\begin{figure}[h]
    \centering
    \input{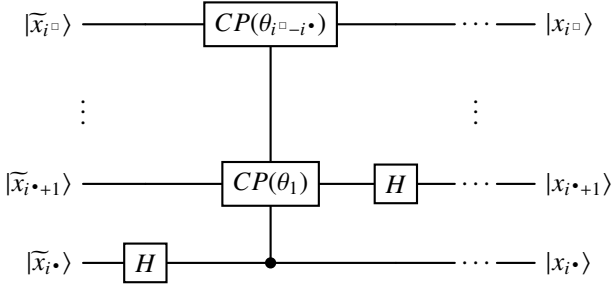}
    \caption{Interaction between the $\qctrl{i}$-th qubit as control and the $\qtarg{i}$-th qubit as target within the iQFT sequence. The diagram highlights the Hadamard gate applied to $\ket{\widetilde{x}_\qctrl{i}}$ and the conditional phase rotation controlled by $\ket{\widetilde{x}_\qctrl{i}}$ acting on $\ket{\widetilde{x}_\qtarg{i}}$.}
    \label{fig:iqft_detail}
\end{figure}

\subsubsection{Gate Complexity and Scaling}
\label{sec:gate_complexity}

The computational cost of the exact iQFT is dominated by the number of two-qubit controlled-phase ($CP$) rotations required to encode the relative phases. For an $n$-qubit register, the standard decomposition involves a sequence of operations where the $j$-th qubit (for $j=1, \dots, n$) acts as a control for $n-j$ target qubits.

The total number of controlled-phase gates, denoted as $N_{CP}$, is determined by the summation of these interactions across the entire circuit:

\begin{equation}
    N_{CP}(n) = \sum_{j=1}^{n-1} j = \frac{n(n-1)}{2}.
    \label{eq:gate_count}
\end{equation}

Consequently, the gate complexity scales quadratically with the system size, i.e., $N_{CP} \in \Theta(n^2)$. While the number of Hadamard gates scales linearly as $\Theta(n)$ and the final SWAP network requires $\Theta(n)$ operations, the quadratic growth of the entangling $CP$ gates presents the primary bottleneck for scalability.

Beyond gate count, the iQFT imposes demanding connectivity constraints on the underlying hardware topology. The algorithm dictates that every qubit $\qctrl{i}$ must interact directly via a two-qubit gate with every other qubit $\qtarg{i}$ (where $\qtarg{i} > \qctrl{i}$) to apply the necessary phase shifts.

From a graph-theoretic perspective, the interaction graph of the iQFT corresponds to a complete graph $K_n$, where an edge exists between every pair of vertices. 

This disparity between the algorithmic requirement (all-to-all) and physical constraints significantly increases the effective circuit depth and the total operation count in realistic implementations.

\subsection{Distributed Quantum Computing: teledata and telegate}
Scaling quantum computation beyond the qubit limits of a single processor requires the interconnection of multiple Quantum Processing Units (QPUs) into a distributed network. In such an architecture, the execution of a monolithic quantum circuit is partitioned across distinct nodes. To realize non-local interactions between qubits residing on physically separated QPUs, the system must rely on quantum communication channels and the distribution of entanglement~\cite{Cirac1999DistributedChannels, Mendes2010SchemesGates,Main2025DistributedLink}.

We assume the network links are capable of generating and sharing Einstein-Podolsky-Rosen (EPR) pairs, specifically the Bell state $|\Phi^+\rangle = \frac{1}{\sqrt{2}}(|00\rangle + |11\rangle)$, between nodes. Utilizing this shared resource, two primary communication primitives are established to support distributed logic: \textit{Teledata}~\cite{Bennett1993TeleportingChannels} and \textit{Telegate}~\cite{Gottesman1999QuantumPrimitive,Gottesman1999DemonstratingOperations}.

\subsubsection{Teledata: state teleportation}
The Teledata strategy, based on the standard quantum teleportation protocol~\cite{Bennett1993TeleportingChannels}, physically relocates the quantum state of a qubit from a source node to a destination node. As depicted in Figure~\ref{fig:teledata_circuit}, this protocol consumes exactly one shared EPR pair and requires one round of classical communication. The fundamental advantage of this primitive is that once the state is reconstructed at the destination, it becomes entirely local; consequently, an arbitrary number of subsequent quantum gates can be executed on the migrated qubit without incurring further inter-node communication penalties. However, if the algorithmic logic demands that the qubit act as a control or target back in its original node later, a ``round-trip'' teleportation is necessary, effectively doubling the entanglement and latency costs.

\subsubsection{Telegate: remote gate application}

\begin{figure}
    \centering
    \input{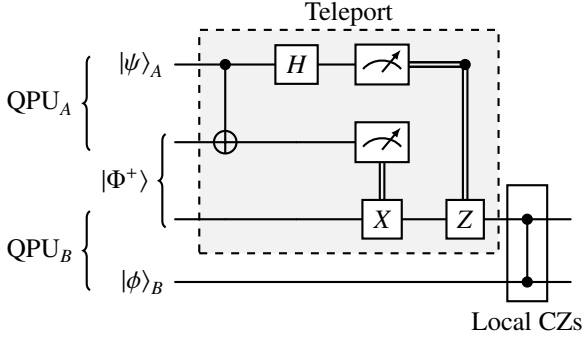}
    \caption{Circuit schematic of the \textit{Teledata} protocol. The quantum state $\ket{\psi_A}$ is teleported from $\text{QPU}_A$ to $\text{QPU}_B$ utilizing a shared EPR pair. Upon receiving the classical measurement outcomes for Pauli corrections ($X, Z$), the state is reconstructed at the destination. This allows the subsequent interaction with $\ket{\phi_B}$ (depicted here as a CZ gate) to be executed as a local operation, avoiding non-local gate overhead.}
    \label{fig:teledata_circuit}
\end{figure}

The Telegate protocol enables the execution of a non-local controlled unitary $U$ between a control qubit in node $A$ and a target qubit in node $B$ without transporting the quantum information of the control qubit itself \cite{Gottesman1999QuantumPrimitive,Gottesman1999DemonstratingOperations}. As illustrated in Figure~\ref{fig:telegate_circuit}, the protocol proceeds in three high-level phases: first, an entanglement stage (Cat-Ent) temporarily extends the superposition of the control across the network to a remote ``proxy'' qubit via a shared EPR pair; second, the intended controlled operation is executed locally at the destination node using the proxy as the control; and finally, a disentanglement stage (Cat-DisEnt) decouples the channel and restores the original control qubit via Pauli corrections. This sequence realizes the remote interaction at a total resource cost of exactly one EPR pair and two bidirectional rounds of classical communication.

\begin{figure}
    \centering
    \input{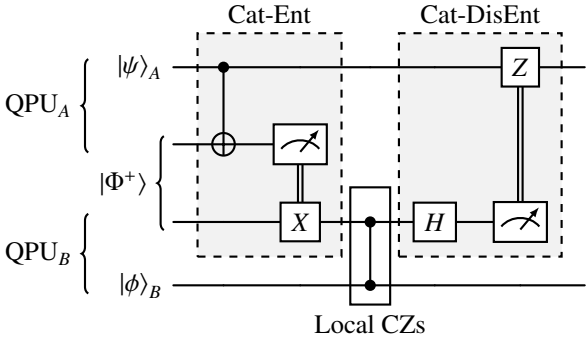}
    \caption{Circuit implementation of the \textit{Telegate} protocol. This primitive executes a non-local controlled-Z gate between a control qubit $\ket{\psi_A}$ in $\text{QPU}_A$ and a target $\ket{\phi_B}$ in $\text{QPU}_B$ using a single shared EPR pair. The process involves an entanglement stage (Cat-Ent) to link the control to the channel, followed by the local interaction, and a disentanglement stage (Cat-DisEnt). Note the bidirectional classical communication required for the $X$ and $Z$ Pauli corrections.}
    \label{fig:telegate_circuit}
\end{figure}

Compared to state teleportation, Telegate is highly resource-efficient for isolated non-local interactions, consuming only half the entanglement required for a Teledata ``round-trip'' (1 EPR pair versus 2). However, this efficiency comes at the expense of higher synchronization overhead. The protocol necessitates a bidirectional exchange of classical information (a signal to initiate the remote sequence and the return of measurement results to finalize the gate) which exposes the circuit to the latency of two distinct classical communication events.

\subsection{Approximate $\operatorname{QFT}^{-1}$ with phase fraction $\varepsilon$}\label{sec:background:error} 

The exact formulation of the $\operatorname{QFT}^{-1}$ entails a complete interaction graph, requiring a two-qubit controlled-phase rotation $CP(\theta_k)$ for every pair of qubits, where $k$ denotes the index difference between the control and target qubits ($\qtarg{i}-\qctrl{i}$). As detailed in Section~\ref{sec:iqft_background}, the rotation angle is given by:

\begin{equation}
    \theta_k = -\frac{\pi}{2^k}.
\end{equation}

A critical observation for optimization is the exponential decay of these rotation angles as $k$ increases. The magnitude of the phase shift is reduced by half for each increment. Consequently, the total phase accumulated by the last target qubit $n-1$, assuming all possible control qubits are in the state $|1\rangle$, follows a geometric progression:

\begin{equation}\label{eq:sum_angles}
    \varphi(n-1) = \sum_{k=1}^{n-1}\theta_{k} = \sum_{k=1}^{n-1} \frac{\pi}{2^k} = \pi \left( 1 - \frac{1}{2^{n-1}} \right).
\end{equation}

As $n \to \infty$, this sum converges to $\pi$, representing a maximum phase shift of $180^\circ$. However, the contributions from distant qubits (large $k$) become vanishingly small. This behavior is illustrated in Figure~\ref{fig:phase_semicircle}, which visualizes how the significance of the rotation diminishes rapidly with $k$.

\begin{figure*}[t]
    \centering
    \begin{subfigure}[b]{0.45\textwidth}
        \centering
        \resizebox{\linewidth}{!}{\begin{tikzpicture}[scale=1]
        \def\R{4.0}

        \draw[gray!20] (-\R-0.2,0) -- (\R+0.2,0); 
        \draw[thick] (\R,0) arc (0:180:\R);       

        \fill[cyan!5] (0,0) -- (0:\R) arc (0:90:\R) -- cycle;
        \draw[cyan!40] (0,0) -- (90:\R);
        \node at (45:0.6*\R) {\normalsize $k=1$};
        \node at (45:0.85*\R) {\normalsize $\theta = \frac{\pi}{2}$};

        \fill[cyan!15] (0,0) -- (90:\R) arc (90:135:\R) -- cycle;
        \draw[cyan!50] (0,0) -- (135:\R);
        \node at (112.5:0.6*\R) {\normalsize $k=2$};
        \node at (112.5:0.85*\R) {\normalsize $\frac{\pi}{4}$};

        \fill[cyan!25] (0,0) -- (135:\R) arc (135:157.5:\R) -- cycle;
        \draw[cyan!60] (0,0) -- (157.5:\R);
        \node[font=\small] at (146.25:0.7*\R) {\normalsize $k=3$};
        \node[font=\small] at (146.25:0.9*\R) {\normalsize $\frac{\pi}{8}$};

        \fill[cyan!35] (0,0) -- (157.5:\R) arc (157.5:168.75:\R) -- cycle;
        \draw[cyan!70] (0,0) -- (168.75:\R);
        \node[font=\tiny, anchor=center] at (163.125:0.7*\R) {\small $k=4$};
        \node[text=black, font=\tiny, anchor=center] at (163.125:0.9*\R) {\normalsize $\frac{\pi}{16}$};

        \fill[cyan!45] (0,0) -- (168.75:\R) arc (168.75:174.375:\R) -- cycle;
        \draw[cyan!80] (0,0) -- (174.375:\R);
        \node[font=\tiny, anchor=center] at (171.5:0.75*\R) {\footnotesize $k=5$};
        \node[text=black, font=\tiny, anchor=center] at (171.5:0.9*\R) {\footnotesize $\frac{\pi}{32}$};

        \fill[black!90] (0,0) -- (174.375:\R) arc (174.375:180:\R) -- cycle;
        \draw[thick] (0,0) -- (180:\R);

        \node[text=white, font=\tiny] at (178:0.8*\R) {$\dots$};
        
        \draw[->, thick, black] (0:0\R+0.3) arc (0:15:\R+0.3);
        \node[anchor=north east, text=black, font=\normalsize] at (10:0\R+0.85) {$\theta_{k}$};

        \node[anchor=north west, font=\normalsize] at (180:\R) {$\pi$};

        \fill (0,0) circle (1.5pt);
        \node[below left] at (0,0) {0};

        \draw[red, ultra thick, dashed] (0,0) -- (174.375:\R+0.1);
        \node[text=black, font=\normalsize] at (174.375:\R+0.3) {$t$};

    \end{tikzpicture}}
        \caption{Geometric representation of phase sectors.}
        \label{fig:phase_semicircle_a}
    \end{subfigure}
    \hfill
    \begin{subfigure}[b]{0.48\textwidth}
        \centering
        \scriptsize
\begin{tikzpicture}
    \begin{semilogyaxis}[
        xlabel={$k$-index},
        ylabel={Rotation angle $\theta_k$ (rad)},
        xmin=1, xmax=15,
        ymin=1e-5, ymax=2,
        grid=both,
        major grid style={dashed, gray!30},
        minor grid style={dotted, gray!20},
        width=0.9\linewidth,
        height=5.2cm,
        legend pos=south west,
        log basis y={10},
        ytickten={-4,-3,-2,-1,0}
    ]
    \addplot[domain=1:15, samples=15, mark=*, blue, thick] {pi * (0.5^x)};
    \addlegendentry{$\theta_k = \pi/2^k$}
    
    \draw[red, dashed, thick] (axis cs:7,1e-5) -- (axis cs:7,2) 
        node[pos=0.85, right] {Threshold $t$};
    \end{semilogyaxis}
\end{tikzpicture}
        \caption{Exponential decay of rotation angles.}
        \label{fig:phase_semicircle_b}
    \end{subfigure}
    
    \caption{Representation of the accumulated phase in the iQFT. \textbf{(a)} Geometric visualization showing the sectors for $k=1$ to $k=5$. The threshold cutoff $t$ indicates the truncation point; interactions beyond this limit contribute a negligible phase shift bounded by $\varepsilon\pi$. \textbf{(b)} Semilogarithmic plot quantifying the rapid exponential decay of the rotation angles $\theta_k = \pi/2^k$ as a function of the $k$ parameter, explicitly highlighting a generic truncation threshold $t =7$.}
    \label{fig:phase_semicircle}
\end{figure*}
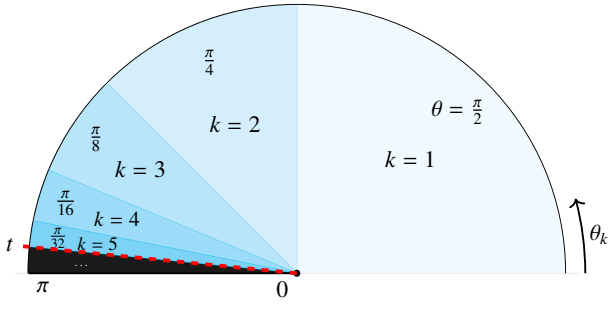
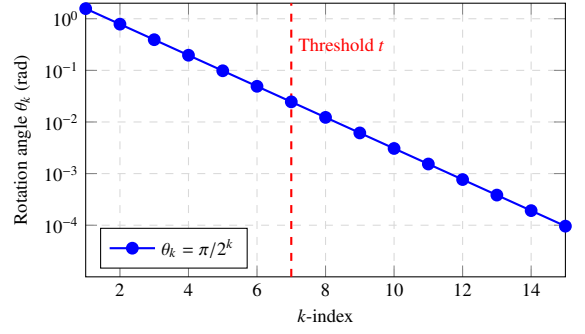

To reduce the gate count, we can employ an approximate $\operatorname{QFT}^{-1}$ by truncating these negligible rotations. To bound the algorithmic error introduced by this truncation, we define a fractional phase tolerance $\varepsilon$, guaranteeing that the total accumulated phase of the discarded gates strictly does not exceed $\varepsilon\pi$. Based on this tolerance, we introduce a threshold parameter $t$, such that all $CP(\theta_k)$ gates with $k > t$ are discarded. The unitary operator for the approximate transform, denoted as $\tilde{U}$, differs from the exact operator $U$ by the omission of these small-angle gates.

To rigorously formalize the algorithmic error introduced by this truncation strategy, we encapsulate the analytical tracking of the discarded rotations into a formal phase-bound framework.

\begin{theorem}[Accumulated Phase Error Bound]
\label{thm:error_bound}
Let $\varphi(n-1) = \sum_{k=1}^{n-1} |\theta_k|$ be the absolute phase accumulated by the final target qubit within an $n$-qubit iQFT register (assuming a worst-case scenario where all controls are $|1\rangle$), and let $\varphi(t) = \sum_{k=1}^{t} |\theta_k|$ be the implemented accumulated phase under a pruning threshold $t$. The absolute phase deviation resulting from omitting all controlled-phase interactions where $k > t$ is strictly upper-bounded by:
\begin{equation}
    \varphi(n-1) - \varphi(t) \leq \frac{\pi}{2^t}
\end{equation}
\end{theorem}

\begin{proof}
By partitioning the exact accumulated phase summation at the threshold cutoff parameter $t$, the residual phase contribution of the truncated controlled-phase gates can be isolated as follows:
\begin{equation}
    \varphi(n-1) = \varphi(t) + \sum_{k=t+1}^{n-1}|\theta_k|.
\end{equation}
Consequently, the deviation between the exact and the implemented rotation is determined entirely by the sum of the truncated angles. By substituting the specific absolute rotation values $|\theta_k| = \pi/2^k$, we obtain:
\begin{equation}\label{eq:error}
    \varphi(n-1) - \varphi(t) = \sum_{k=t+1}^{n-1}\frac{\pi}{2^k} = \pi  \sum_{k=t+1}^{n-1} 2^{-k} .
\end{equation}
Using the geometric series property $\sum_{k=t+1}^{n-1} 2^{-k} = 2^{-t} - 2^{-(n-1)}$, the analytical expression for the approximation error evaluates exactly to:
\begin{equation}
    \varphi(n-1) - \varphi(t) = \pi\left(\frac{1}{2^t} - \frac{1}{2^{n-1}}\right).
\end{equation}
In the asymptotic limit as the global register scales ($n \to \infty$), the negative component $1/2^{n-1}$ vanishes. This establishes the strict, network-independent upper bound:
\begin{equation}
    \varphi(n-1) - \varphi(t) \leq \frac{\pi}{2^t}.
\end{equation}
\end{proof}

To ensure that the approximation does not exceed the specified tolerance $\varepsilon\pi$, we require the bounded deviation to satisfy $\varphi(n-1) - \varphi(t) \leq \varepsilon\pi$. Solving for the threshold $t$, we obtain the cutoff condition:
\begin{equation}
    2^t \geq \frac{\pi}{\varepsilon\pi} \implies t \geq \log_2 \left( \frac{1}{\varepsilon} \right) = -\log_2(\varepsilon).
\end{equation}

Thus, to maintain an accuracy of $\varepsilon$, it is sufficient to retain only those controlled-phase gates where 
\begin{equation}\label{eq:threshold_definition}
k \leq t = \lceil -\log_2(\varepsilon) \rceil. 
\end{equation}

This transforms the gate complexity from $\mathcal{O}(n^2)$ to $\mathcal{O}(n\cdot t)$, offering a significant reduction for large $n$ while preserving algorithmic fidelity.

\section{Methods and methodology}\label{sec:methods}
\subsection{Distributed iQFT across $P$ nodes}\label{sec:methods:distributed_iqft}

\begin{figure*}
    \centering
    \resizebox{\linewidth}{!}{\input{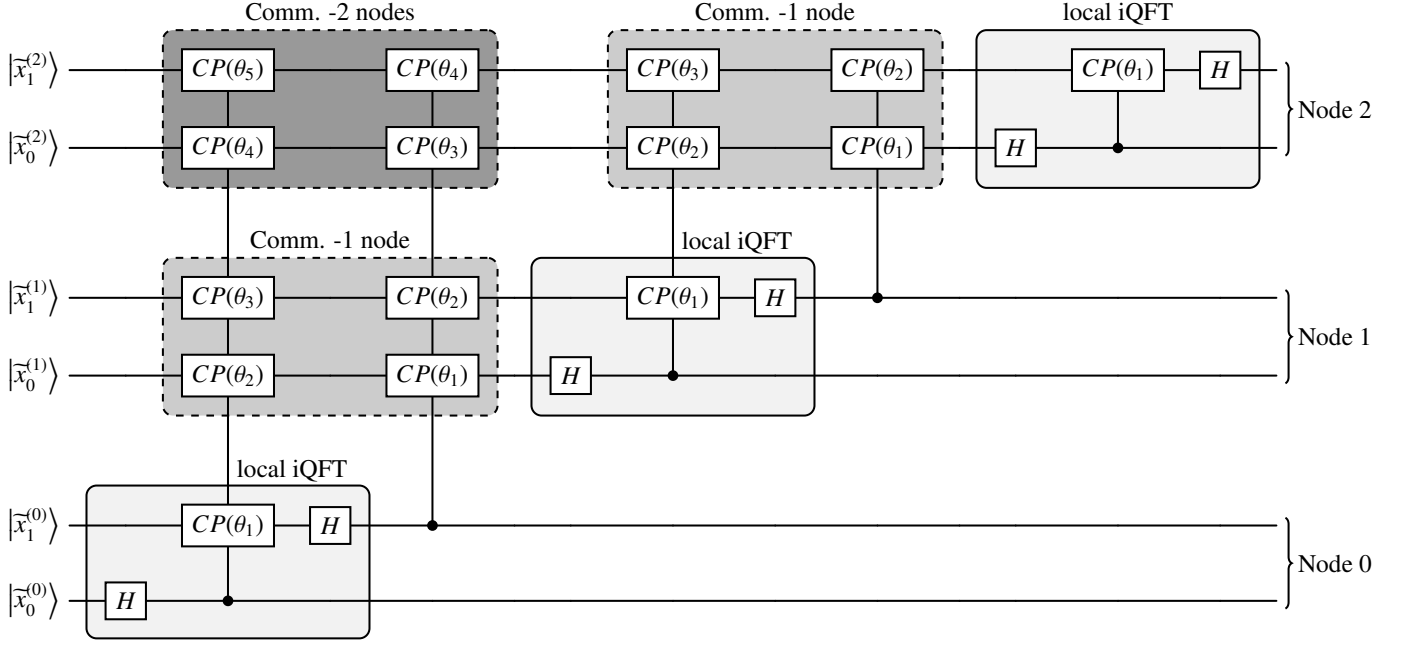}}
    
    \caption{Distributed decomposition of a 6-qubit iQFT across 3 nodes ($p=0,1,2$). The circuit is restructured into Local iQFT blocks (light gray), which require no external communication, and Communication Blocks (darker gray), where Node $p$ interacts with qubits from previous nodes $p' < p$. The phase angles follow $\theta_k = \pi/2^k$.}
    \label{fig:dist_iqft_3nodes}
\end{figure*}

To distribute the \gls{IQFT} across a networked architecture, we must restructure the monolithic circuit into a sequence of operations that distinguishes between intra-node (local) and inter-node (distributed) dependencies.


To formalize the distributed architecture, we define a quantum circuit characterized by the following parameters:

\begin{itemize}
    \item \textbf{Network dimension:} The system consists of $P$ nodes, indexed by the identifier $p \in \{0, 1, \dots, P-1\}$. To explicitly distinguish functional roles during distributed quantum operations, we introduce the modifiers $\qctrl{p}$ to denote a node acting as a control and $\qtarg{p}$ to denote a target node.
    
    \item \textbf{Node capacity:} We assume a homogeneous distribution where each node hosts exactly $Q$ local qubits. Consequently, the total global register size is $n = P \cdot Q$.
    
    \item \textbf{Qubit state notation:} We introduce the notation $|\tilde{x}_q^{(p)}\rangle$ to denote the state of the $q$-th qubit residing in the $p$-th node. Here, $q \in \{0, 1, \dots, Q-1\}$ represents the local index of the qubit within its host node. Similar to the node notation, we use $\qctrl{q}$ and $\qtarg{q}$ to designate a control and target qubit, respectively. To simplify complex mathematical expressions and improve readability, we define the shorthand notation $\qctrl{|\tilde{x}_q^{(p)}\rangle}$ as strictly equivalent to $|\tilde{x}_{\qctrl{q}}^{(\qctrl{p})}\rangle$, concisely indicating that both the specific qubit and its host node serve the control role. The same equivalence applies to the target notation $\qtarg{|\tilde{x}_q^{(p)}\rangle} \equiv |\tilde{x}_{\qtarg{q}}^{(\qtarg{p})}\rangle$.
    
    \item \textbf{Global mapping:} The global index of any given qubit within the complete $n$-qubit register can be recovered through the linear mapping function $I(q,p) = p \cdot Q + q$.
\end{itemize}

The standard iQFT decomposition (as seen in Section~\ref{sec:iqft_background}) requires applying Controlled-Phase ($CP$) rotations between every pair of qubits. In a distributed setting, these interactions fall into two distinct categories:
\begin{enumerate}
    \item \textbf{Local blocks:} Interactions where both the control and target qubits reside within the same node $p$ (so, $\qtarg{p}=\qctrl{p}$). These operations can be executed entirely within the QPU without utilizing the quantum network.
    \item \textbf{Communication blocks:} Interactions where the control qubit resides in node $\qctrl{p}$ and the target qubit resides in node $\qtarg{p}$ (with $\qtarg{p}\neq\qctrl{p}$). These require entanglement-assisted communication primitives to realize the non-local phase shifts.
\end{enumerate}

Figure~\ref{fig:dist_iqft_3nodes} illustrates this segmented architecture for a specific case of $P = 3$ nodes with $Q = 2$ qubits each ($n = 6$). The structure reveals a distinct execution pattern: each node $p$ must first apply a series of \textit{Communication blocks}, interacting sequentially with nodes $\qctrl{p} < p$. Only after these non-local phases are applied can the node proceed to execute its \textit{Local block}. Crucially, this local component can be identified as a lower-order iQFT instance; specifically, it corresponds to a local iQFT acting on the $Q$ qubits assigned to the node, executed without external dependencies.

The generalized interaction logic is depicted in Figure~\ref{fig:generic_interaction}. For a control qubit $\qctrl{|\tilde{x}_i^{(p)}\rangle}$ in node $\qctrl{p}$ and a target qubit $\qtarg{|\tilde{x}_i^{(p)}\rangle}$  in a remote node $\qtarg{p}$, the rotation angle depends on the global index difference between them. Assuming the node index $\qtarg{p}$ represents the higher-order significance bits, the index difference $k$ between the control qubit $\qctrl{q}$ of node $\qctrl{p}$ and the target qubit $\qtarg{q}$ of node $\qtarg{p}$ is derived from their global indices:
\begin{equation}
    \qctrl{I} = \qctrl{p} \cdot Q + \qctrl{q}, \quad \qtarg{I} = \qtarg{p} \cdot Q + \qtarg{j}.
\end{equation}

The rotation index $k$ corresponds to the difference $|\qtarg{I} - \qctrl{I}|$. In the standard decomposition where controls are applied from distant nodes $p < p'$, the specific rotation angle appliying Eq.~\ref{eq:theta_generic_iqft} is given by:

\begin{equation}\label{eq:theta_generic_diqft}
    \theta(\qctrl{I}, \qtarg{I}) = \frac{-\pi}{2^{Q(\qtarg{p}-\qctrl{p}) + (\qtarg{q}-\qctrl{q})}}.
\end{equation}

This formulation allows us to map the abstract algorithm directly to physical link requirements. The communication blocks effectively group these interactions.

\begin{figure}
    \centering
    \input{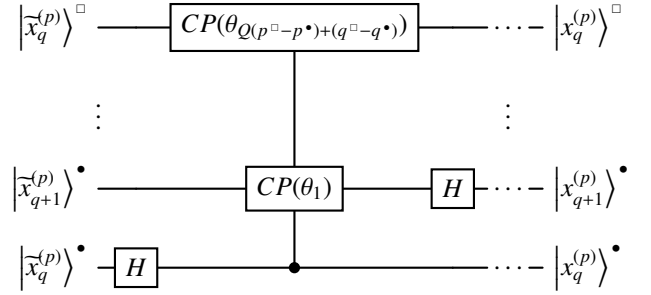}
    \caption{Generic interaction between a target qubit $\qtarg{q}$ in node $\qtarg{p}$ and a control qubit $\qctrl{q}$ in a distant node $\qtarg{p}$. The rotation angle is determined by the global index difference, formalized as $\theta_k$ where $k = \frac{-\pi}{2^{k}}$.}
    \label{fig:generic_interaction}
\end{figure}

To formally demonstrate that the intra-node operations constitute a standard iQFT of dimension $Q$ (local block), we apply the locality condition $\qctrl{p} = \qtarg{p}$ to the generalized rotation formula derived in Eq.~\ref{eq:theta_generic_diqft}. When both control and target qubits reside in the same node, the term $Q(\qtarg{p}-\qctrl{p}) = 0$. Consequently, the rotation angle $\theta$ simplifies to depend solely on the local indices $\qctrl{q}$ and $\qtarg{q}$:

\begin{equation}
    \theta_{\qtarg{p}=\qctrl{p}} = \frac{-\pi}{2^{Q(0) + (\qtarg{q}-\qctrl{q})}} = \frac{-\pi}{2^{\qtarg{q}-\qctrl{q}}}, \quad \text{for } \qtarg{q} > \qctrl{q}.
\end{equation}

This expression is mathematically identical to the phase shift definition of the monolithic iQFT shown in Figure~\ref{fig:generic_interaction} and Eq.~\ref{eq:theta_generic_iqft}, specifically for a difference $k = \qtarg{i}-\qctrl{i}$. Since the indices $i$ range from $0$ to $Q-1$, the set of operations within this block generates exactly the interaction graph of a $Q$-qubit iQFT. Thus, the \textit{Local Block} is equivalent to an iQFT($Q$), verifying that no inter-node communication is required for these specific gates.

Figure~\ref{fig:diqft_restructured} depicts the generalized execution schedule for an arbitrary node $p$ within the network. The quantum circuit is abstracted into a sequence of communicative layers followed by a local computation (as shown in Figure~\ref{fig:dist_iqft_3nodes}).

Instead of treating non-local interactions as isolated gates, the interaction with preceding nodes is structured into communicative blocks denoted by $\mathbf{CP}(\Theta_d)$. Here, $d$ represents the logical distance to the control node (i.e., interacting with node $\qtarg{p} = \qctrl{p}+d$). Each block encapsulates the complete set of remote controlled-phase rotations required between the $Q$ qubits of the target node $\qtarg{p}$ and the $Q$ qubits of the control node $\qtarg{p}-d = \qctrl{p}$.

Mathematically, $\Theta_d$ is defined as the set of all phase rotations to be applied within that specific block. Following the generalized rotation derived in Eq.~\ref{eq:theta_generic_diqft}, this set is strictly formulated as:
\begin{equation}\label{eq:communication_block}
    \Theta_d = \left\{ \theta\left(I\left(\qctrl{q},\qctrl{p}\right),I\left(\qtarg{q},\qtarg{p}\right)\right) \mid \qtarg{p}=\qctrl{p}+d\right\},
\end{equation}
where $\qtarg{q}$ denotes the qubit index in the target register (node $\qtarg{p}$) and $\qctrl{q}$ the index in the control register (node $\qtarg{p}-d=\qctrl{p}$). This formulation ensures that all cross-node dependencies are satisfied for the given distance.

As shown in the circuit, the sequence proceeds from the most distant interactions ($d=p$, corresponding to node 0) to the nearest neighbor ($d=1$, corresponding to node $p-1$).\\

\subsubsection{Entanglement Resource Allocation Profile} 

Entanglement cost is strictly tied to the inter-node migration of the quantum state, rather than the number of remote gates. When a control qubit $|\tilde{x}_{q}^{(p)}\rangle^\bullet$ must interact with multiple targets in a remote node $p^\Box$, the Telegate protocol consumes exactly one EPR pair to create a cat-entangled proxy at the destination. Once this proxy is established in the remote node, it can be applied to an arbitrary number of target qubits with zero additional entanglement cost.

\begin{figure}
    \centering
    \input{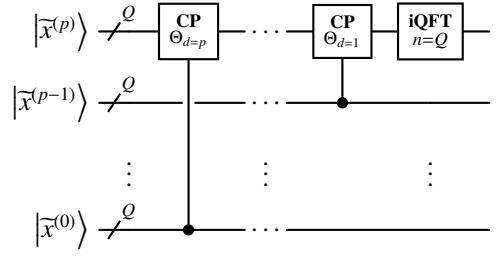}
    \caption{Generic interaction between nodes, execution model for an arbitrary node $p$. The circuit illustrates the sequential application of communication blocks $\mathbf{CP}$, starting from the most distant dependency (Node 0, $d=p$) up to the nearest (Node $p-1$, $d=1$). The execution of this sequence assumes that all previous nodes $0, \dots, p-1$ have already initiated their corresponding operations to act as controls. The process concludes with a local inverse QFT on the $Q$ qubits of node $p$.}    
    \label{fig:diqft_restructured}
\end{figure}

\subsection{Communication Efficiency: block pruning strategy}\label{sec:methods:communication_efficiency}

While the standard approximate iQFT reduces the total gate count by discarding small-angle rotations, its impact is particularly transformative in a distributed setting. In monolithic architectures, removing a gate merely saves execution time. However, in a distributed architecture, the primary bottlenecks are the generation of entanglement (EPR pairs) and the latency of classical communication rounds. Therefore, the optimization goal shifts from simply reducing $N_{CP}$ to minimizing the number of active inter-node links.

We leverage the geometric decay of the rotation angles derived in Section~\ref{sec:iqft_background} to define a \textit{communication horizon}. Recall that for a target accuracy $\varepsilon$, any controlled-phase rotation $CP(\theta_k)$ with distance $k > t$ can be approximated as the identity, where $t = \lceil -\log_2(\varepsilon) \rceil$ (Eq.~\ref{eq:threshold_definition}). We apply this threshold to the coarse-grained communication blocks $\mathbf{CP}(\Theta_d)$ (Eq.~\ref{eq:communication_block}) defined in Section~\ref{sec:methods:distributed_iqft}.

Based on this truncation condition, we can formalize the maximum interaction range required in the distributed architecture.

\begin{proposition}[Communication Horizon]
\label{prop:comm_horizon}
For a distributed quantum register of node capacity $Q$ and a target algorithmic tolerance $\varepsilon$, the maximum required interaction distance between nodes, denoted as the communication horizon $d_\text{max}$, is strictly bounded by:
\begin{equation}\label{eq:communication_horizon}
    d_\text{max} = \left\lfloor \frac{\lceil -\log_2(\varepsilon) \rceil - 1}{Q} \right\rfloor + 1.
\end{equation}
\end{proposition}

\begin{proof}
A communication block $\mathbf{CP}(\Theta_d)$ encapsulates interactions between a target node $\qtarg{p}$ and a control node at distance $d$. To determine the relevance of the entire block, we consider the \textit{maximum interaction strength} within it, which corresponds to the smallest qubit distance $k_{min}$ (i.e., the largest rotation angle). This occurs between the least significant qubit of the target node ($\qtarg{q}=0$) and the most significant qubit of the control node ($\qctrl{q}=Q-1$). Applying Eq.~\ref{eq:theta_generic_diqft}:
\begin{equation}
    k_{min}(d) = Q \cdot (\qtarg{p}-\qctrl{p}) + (0) - (Q-1) = Q(d - 1) + 1.
\end{equation}
If this strongest interaction falls below the significance threshold (i.e., if $k_{min}(d) > t$), then every other interaction in the block will also satisfy $k > t$. Consequently, the entire communication block becomes negligible and can be pruned from the schedule. Setting the retention condition $k_{min}(d_\text{max}) \le t$, we solve for $d_\text{max}$:
\begin{equation}
    Q(d_\text{max} - 1) + 1 \le t \implies d_\text{max} \le \frac{t - 1}{Q} + 1.
\end{equation}
Substituting $t$ from Theorem \ref{thm:error_bound} (Eq.~\ref{eq:threshold_definition}) and applying the floor function for integer node distances completes the proof.
\end{proof}

This structural limitation directly alters the asymptotic scaling of the algorithm's quantum network footprint.

\begin{corollary}[Linear Global Entanglement Scaling]
\label{cor:scaling}
In the communication efficient distributed iQFT, the global entanglement complexity scales linearly $\mathcal{O}(P)$ with respect to the total number of nodes $P$, rather than the quadratic $\mathcal{O}(P^2)$ scaling of the exact distributed implementation.
\end{corollary}

\begin{proof}
In the exact distributed iQFT architecture, each node must establish entanglement links with all its preceding nodes, yielding an all-to-all communication topology where the total number of inter-node blocks grows as $\sum_{i=1}^{P-1} i \in \mathcal{O}(P^2)$. By Proposition \ref{prop:comm_horizon}, the pruning strategy bounds the connectivity degree of any node $p$ to a maximum of $d_\text{max}$ nearest neighbors. Since $d_\text{max}$ depends solely on the localized parameters $Q$ and $\varepsilon$ and is independent of $P$, the maximum number of required communication links per node is bounded by a constant $\mathcal{O}(1)$. Summing this over all $P$ nodes reduces the aggregate global communication complexity to $\mathcal{O}(P)$.
\end{proof}

\subsubsection{Granularity of the Truncation Strategy}

Pruning operates strictly at the block level. If a block is retained ($k_{min}(d) \le t$), all its internal gates are executed—even those where $k > t$. Since the EPR pair required to migrate the control qubit is already consumed by the block's significant interactions, executing the remaining weak rotations incurs zero additional communication cost. Retaining them maximizes algorithmic fidelity without any entanglement overhead.

To illustrate the practical impact, consider a network composed of nodes with capacity $Q=4$ qubits. Targeting a high algorithmic fidelity with $\varepsilon = 10^{-5}$, the phase truncation threshold is $t \approx 17$. Applying Eq.~\ref{eq:communication_horizon}, the communication horizon is bounded by $d_\text{max} = 5$. This implies that in a large-scale network (e.g., 100 nodes), Node 99 is not required to establish costly long-range entanglement links with Nodes 0 through 93; it strictly needs to coordinate with its 5 nearest predecessors (Nodes 94--98). Interactions beyond this horizon contribute negligible phase shifts and are safely discarded, drastically reducing the demand on quantum repeater networks.

\subsection{Metrics evaluated}
To rigorously assess the impact of the proposed distributed architecture and optimization strategies, we quantify performance using the following set of metrics. These metrics are designed to evaluate the trade-offs between algorithmic fidelity, communication resources, and circuit locality.

\subsubsection{Fidelity vs. Threshold}
We evaluate the accuracy of the approximate distributed iQFT by computing the state fidelity, defined as:

\begin{equation}\label{eq:fidelity}
    \mathcal{F} = |\langle \psi_{ideal} | \psi_{approx} \rangle|^2, 
\end{equation}
\noindent where $|\psi_{ideal}\rangle$ is the output state of the exact transform and $|\psi_{approx}\rangle$ is the output of the pruned circuit. This metric is analyzed as a function of the pruning threshold $t$. The objective is to validate the theoretical error model derived in Section~\ref{sec:background:error} by comparing the fidelity obtained from full state-vector simulations against the theoretical lower bound implied by the chosen accuracy parameter $\varepsilon$. \\

\noindent\textbf{What does it prove?} This serves to empirically validate the theoretical error bound and cutoff condition derived in Section~\ref{sec:background:error} (Eq.~\ref{eq:threshold_definition}).

\subsubsection{Expected State Probability vs. Threshold}

While state-vector simulations provide exact fidelity calculations under ideal conditions, assessing the architecture's viability in the NISQ era requires modeling stochastic hardware errors. To facilitate the introduction of realistic physical noise profiles, such as those provided by IBM Qiskit Runtime FakeBackends, we utilize a measurement-based error metric compatible with shot-based simulators (which output discrete measurement counts rather than exact state amplitudes).

The evaluation protocol involves initializing the system in a known Fourier basis state $|\tilde{x}\rangle$, which corresponds to the ideal QFT transformation of a computational basis state $|x\rangle$. Upon applying the approximate distributed iQFT, we measure the output in the computational basis. In an ideal execution, the probability mass should concentrate entirely on the state $|x\rangle$. We define the empirical error $E_\text{counts}$ as the fraction of measurement shots that fail to collapse into the expected computational state:
\begin{equation}
    E_\text{counts} = 1 - \frac{N_{x}}{N_\text{total}},
\end{equation}
where $N_{x}$ is the count of shots matching the expected bitstring $x$, and $N_\text{total}$ is the total number of shots executed. \\

\noindent\textbf{What does it prove?} This metric allows us to evaluate how the execution of the circuit is affected by applying different pruning thresholds in noisy systems, providing a practical method to assess algorithmic resilience in large-scale distributed architectures where direct fidelity computation is intractable.

\subsubsection{Scalability of entanglement resources (EPR per Node)}
This metric quantifies the aggregate cost of inter-node communication required to execute the algorithm. We measure the total count of EPR pairs consumed across the entire network during the circuit execution. We analyze the behavior of this cost as a function of the pruning threshold $t$, comparing the baseline distributed iQFT against the optimized implementation. \\

\noindent\textbf{What does it prove?} This metric directly demonstrates the global communication savings achieved by the pruning strategy described in Section~\ref{sec:methods:communication_efficiency}.

\subsubsection{Accuracy under connectivity constraints}
Suppose hardware limitation of a fixed communication horizon $d_\text{max}$ (the maximum distance between interacting nodes) and a specific node capacity $Q$, we analyze the achievable algorithmic accuracy. By inverting the relationship established in the block pruning strategy, we determine the minimum error $\varepsilon$ that can be guaranteed when the communication topology is restricted. \\

\noindent\textbf{What does it prove?} This analysis confirms the feasibility of the bounded interaction model and the communication horizon ($d_\text{max}$) proposed in Section~\ref{sec:methods:communication_efficiency}.

\subsubsection{Coupling Ratio}
We define the Coupling Ratio $\eta$ as the quotient between the number of remote controlled-phase gates (inter-node) and local controlled-phase gates (intra-node). This metric serves as a proxy for the degree of independence of the nodes. We analyze how $\eta$ evolves with respect to the node capacity $Q$, the total network size $P$, and the approximation tolerance $\varepsilon$. \\

\begin{equation}\label{eq:eta_locality_ratio}
    \eta(Q, P, \varepsilon) = \frac{N_{CP}^\text{remote}}{N_{CP}^\text{local}}
\end{equation}

The objective is to try to achieve a $\eta\rightarrow0$ in the distributed algorithm. \\

\noindent\textbf{What does it prove?} This demonstrates the effectiveness of the distributed decomposition in Section~\ref{sec:methods:distributed_iqft} in clustering operations locally to minimize network latency.

\subsubsection{Communication overhead}
To quantify the operational penalty introduced by the network protocols, we define the \textit{Communication Overhead} ($\gamma$). This metric is calculated as the ratio between the number of physical operations dedicated exclusively to quantum communication ($N_{\text{comm}}$) and the number of logical computational gates ($N_{\text{comp}}$):

\begin{equation}
    \gamma = \frac{N_{\text{comm}}}{N_{\text{comp}}}
\end{equation}

In the context of the \textit{Telegate} protocol utilized in this work, $N_{\text{comm}}$ encompasses all auxiliary operations necessary for state transfer, including Bell pair generation, entanglement swapping operations (CNOTs), and measurement-dependent corrections. Conversely, $N_{\text{comp}}$ refers strictly to the logical controlled-phase rotations required by the iQFT algorithm. A visual decomposition distinguishing between these communication-specific overheads and the fundamental computational gates is presented in Figure~\ref{fig:telegate_circuit}. It is important to emphasize that $\gamma$ is a global metric, representing the aggregate overhead of the distributed execution rather than a per-node average. \\

\noindent\textbf{What does it prove?} This metric exposes the communication cost of the distributed implementation, quantifying the exact ratio of physical overhead operations required to execute a single useful logical gate.

\subsection{Experimental methodology}

To validate the proposed distributed architecture and the efficiency of the pruning strategy, we implemented the distributed iQFT algorithm utilizing the Qiskit software development kit. The numerical simulations were performed using the statevector simulator backend provided by Qiskit Aer. This simulator performs an ideal execution of the quantum circuit, computing the exact final quantum state vector $|\psi_{final}\rangle$ rather than a probability distribution of measurement outcomes.

The choice of a state-vector simulation is essential for this study for two primary reasons: first, it allows for the precise calculation of the algorithmic fidelity $\mathcal{F}$ by computing the overlap with the ideal mathematical output, enabling a direct verification of the theoretical error bounds and threshold equations derived in Section~\ref{sec:background:error}. Second, it permits the exact tracking of entanglement resources without the statistical noise associated with shot-based sampling.

All experiments were conducted on the High-Performance Computing cluster of the \textit{Centro Singular de Investigación en Tecnoloxías Intelixentes} (CiTIUS)\footnote{More information about the HPC cluster: \url{https://wiki.citius.gal/centro:servizos:hpc}}. Specifically, the simulations were executed on a dedicated compute node designed for memory-intensive workloads, equipped with the following specifications:
\begin{itemize}
    \item \textbf{Processors:} 4 $\times$ Intel Xeon Gold 6248 CPUs @ 2.50GHz (20 cores per socket, totaling 80 physical cores).
    \item \textbf{Memory:} 1 TB of RAM.
\end{itemize}
This high-memory configuration is critical for the study, as the memory requirements for state-vector simulations scale exponentially with the number of qubits ($16 \times 2^n$ bytes). The available 1 TB of RAM enabled the simulation of distributed systems with sufficient qubits to analyze the asymptotic behavior of the communication horizon $d_\text{max}$ and the locality ratio $\eta$.
\section{Results and Discussion}\label{sec:results_discussion}

The results presented in this work demonstrate that the communication aware optimization of the distributed iQFT significantly improves scalability without compromising algorithmic integrity. By restructuring the circuit to exploit the quadratic decay of controlled-phase rotation angles, we effectively transform a global, all-to-all interaction problem into a localized execution model. In this section, we analyze the performance of the proposed architecture, evaluating the trade-offs between algorithmic fidelity, communication resources, and circuit locality.

\subsection{Fidelity and Theoretical Bounds}

\begin{figure}[h]
    \centering
    \scriptsize
\begin{tikzpicture}
    \begin{axis}[
        width=8cm, height=7cm,
        xlabel={Threshold $t$},
        ylabel={Infidelity ($1 - \mathcal{F}$)},
        ymode=log, 
        grid=both,
        grid style={line width=.1pt, draw=gray!20},
        major grid style={line width=.2pt,draw=gray!50},
        xmin=2, xmax=16, 
        xtick={2,4,6,8,10,12,14,16},
        legend pos=south west,
        legend style={font=\footnotesize},
    ]
        \addplot[
            color=blue,
            mark=*,
            mark size=2pt,
            thick
        ]
        table [
            x=limit,       
            y=qubits_18,   
            col sep=comma  
        ] {tikz/results/csv/errors.csv};
        \addlegendentry{Simulated ($Q=18$)}
        \addplot[
            color=red,
            mark=square*,
            mark size=2pt,
            style=dashed,
            thick
        ]
        table [
            x=limit, 
            y=qubits_18, 
            col sep=comma
        ] {tikz/results/csv/theoretical_errors.csv};
        \addlegendentry{Theoretical}
    \end{axis}
\end{tikzpicture}
    \caption{Analysis of the algorithmic infidelity ($1 - \mathcal{F}$) as a function of the pruning threshold $t$ for local register size $Q=18$. The data is plotted on a semilogarithmic scale to visualize the \textbf{exponential decay} of the error. Solid lines represent the results obtained from full state-vector simulations, while dashed lines correspond to the theoretical error bounds derived in Eq.~\ref{eq:threshold_definition}.}   
    \label{fig:results:fidelity}
\end{figure}
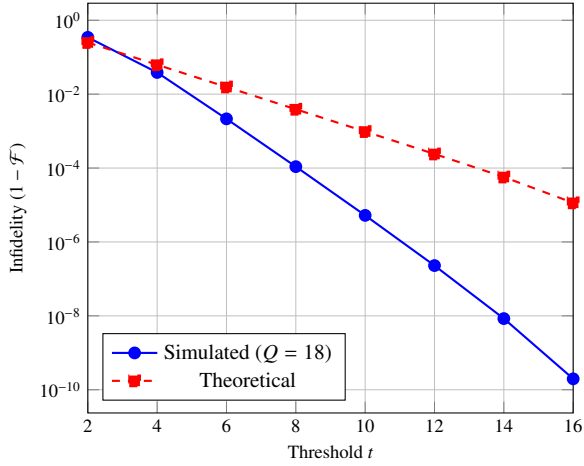

Figure~\ref{fig:results:fidelity} illustrates the infidelity, defined as $1 - \mathcal{F}$ (Eq.~\ref{eq:fidelity}), as a function of the pruning threshold $t$ (Eq.~\ref{eq:threshold_definition}) for a system size of $Q=18$ qubits. The data is presented on a semilogarithmic scale to accommodate the wide dynamic range of the error, which spans from approximately $10^{-1}$ to $10^{-10}$.

The linearity of the curves indicates that the infidelity decays exponentially with a linear increase in the threshold $t$. This behavior confirms the efficacy of the proposed pruning strategy: a moderate increase in the threshold parameter yields a reduction in the final state error by several orders of magnitude. For instance, in the $Q=18$ case, increasing the threshold from $t=4$ to $t=10$ suppresses the error by nearly four orders of magnitude.

A consistent divergence is observed between the numerical simulations and the analytical predictions, where the simulated infidelity (solid lines) remains strictly bounded by the theoretical error estimates (dashed lines). This discrepancy is attributable to the conservative nature of the theoretical bound, which is derived under a worst-case scenario assumption. The theoretical upper model assumes that every pruned controlled-phase gate, $CP(\theta)$, introduces its maximal possible error (physically corresponding to the control and target qubits invariably being in the state $|1\rangle$). However, in a practical execution environment, the quantum register exists in a superposition of computational basis states. For any subspace of the wavefunction where the control qubit is $|0\rangle$, the $CP(\theta)$ gate acts as the identity operator ($\mathbb{I}$), and its truncation introduces zero error. Consequently, the actual accumulated phase error is statistically lower than the sum of the absolute magnitudes of the discarded angles, offering a safety margin where the realized algorithmic fidelity exceeds the strict design requirements.

\subsection{Expected State Probability under Fake Backends}
\label{sec:results:expected_state_probability}

To evaluate the resilience of the proposed architecture against physical decoherence, we analyzed the Expected State Probability ($E_\text{counts}$) across different truncation thresholds $t$. The algorithm was executed on an ideal state-vector simulator alongside three realistic noisy hardware models provided by IBM Qiskit Runtime (Fake Hanoi, Fake Montreal, and Fake Paris). The results are depicted in Figure~\ref{fig:expected_probability}.

\begin{figure}
    \centering
    \scriptsize
\begin{tikzpicture}
\begin{axis}[
    xlabel={Threshold \textit{t}},
    ylabel={Expected State Probability},
    xmin=2, xmax=16,
    ymin=0, ymax=1.05,
    xtick={2,4,6,8,10,12,14,16}, 
    ymajorgrids=true,
    grid=both,
    grid style={line width=.1pt, draw=gray!20},
    major grid style={line width=.2pt,draw=gray!50},
    legend style={
        at={(0.55, 0.55)},
        anchor=west,
        font=\footnotesize
    },
    width=8cm, height=7cm,
    mark size=2pt
]

\addplot+ [color=blue, thick, mark=*] table [x=Limit, y=fake_hanoi, col sep=comma] {tikz/results/csv/datos_probabilidad.csv};
\addlegendentry{Fake Hanoi}
\addplot [color=blue, dashed, thick, domain=2:16, samples=2, forget plot] {0.1796875};

\addplot+ [color=red, thick, mark=square*] table [x=Limit, y=fake_montreal, col sep=comma] {tikz/results/csv/datos_probabilidad.csv};
\addlegendentry{Fake Montreal}
\addplot [color=red, dashed, thick, domain=2:16, samples=2, forget plot] {0.1298828125};

\addplot+ [color=brown, thick, mark=triangle*] table [x=Limit, y=fake_paris, col sep=comma] {tikz/results/csv/datos_probabilidad.csv};
\addlegendentry{Fake Paris}
\addplot [color=brown, dashed, thick, domain=2:16, samples=2, forget plot] {0.0322265625};

\addplot+ [color=black, thick, mark=diamond*] table [x=Limit, y=Ideal_Simulator, col sep=comma] {tikz/results/csv/datos_probabilidad.csv};
\addlegendentry{Ideal Simulator}
\addplot [color=black, dashed, thick, domain=2:16, samples=2, forget plot] {1.0};

\end{axis}
\end{tikzpicture}
    \caption{Expected State Probability ($1-E_\text{counts}$) as a function of the truncation threshold $t$ for register $Q = 18$. The ideal simulator converges to 1.0, whereas noisy backends degrade due to physical circuit depth. The horizontal dashed lines represent the baseline performance of the standard iQFT without thresholding for each respective backend.}
    \label{fig:expected_probability}
\end{figure}
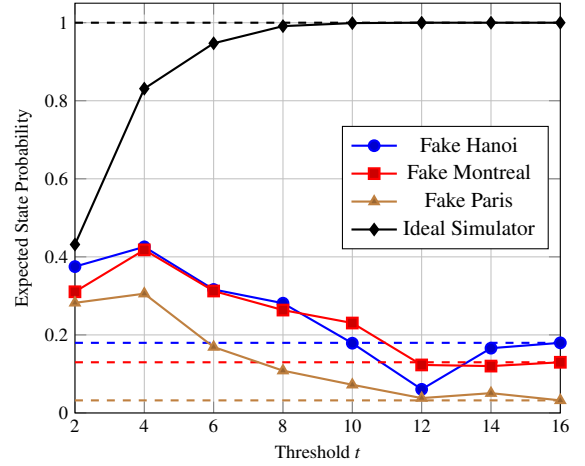

For the ideal simulator, increasing the threshold $t$ monotonically improves the success probability. As $t$ grows, the algorithm transitions from a highly approximate state to a more exact transformation, reaching a 100\% success rate at $t=16$. Interestingly, at the most aggressive pruning level ($t=2$), the algorithmic error dominates the execution, yielding an ideal baseline probability that is nearly identical to those obtained from the noisy hardware models.

Conversely, the noisy backends exhibit a markedly inverse trend: as the threshold $t$ increases, the empirical success probability steadily degrades (i.e., the error $E_{\text{counts}}$ increases). This degradation is fundamentally driven by the severe mismatch between the algorithmic requirements of the distributed iQFT and the physical connectivity of the hardware.

\begin{figure}
    \centering
    \includegraphics[width=\linewidth]{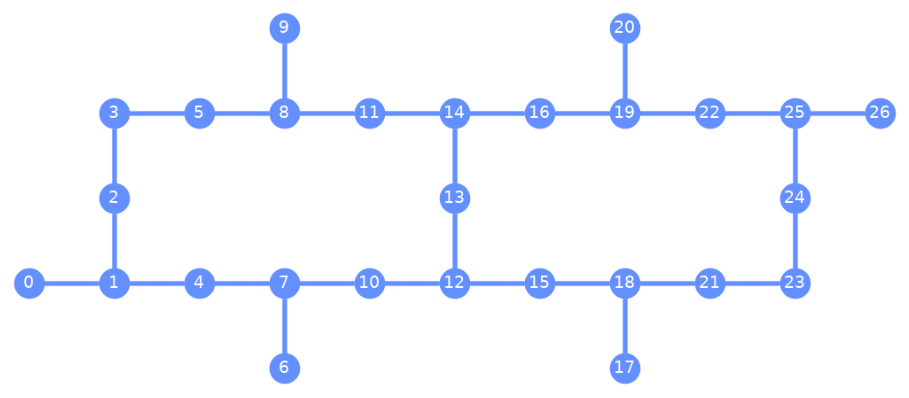}
    \caption{Topological coupling map of the three IBM Qiskit Fake backends utilized in the simulations (Hanoi, Montreal and Paris). The sparse, restricted connectivity imposes a heavy SWAP-gate penalty for two-qubit operations.}
    \label{fig:coupling_map}
\end{figure}

\begin{figure*}
    \centering
    \scriptsize
\begin{tikzpicture}
    \pgfplotsset{
        colormap={fire}{
            color(0cm)=(white);
            color(0.25cm)=(yellow!50!white);
            color(0.5cm)=(orange);
            color(0.75cm)=(red!90!orange);
            color(1cm)=(red!60!black)
        }
    }
    \begin{groupplot}[
        group style={
            group size=3 by 1,  
            horizontal sep=0.25cm, 
            x descriptions at=edge bottom, 
            y descriptions at=edge left,   
        },
        width=0.275\linewidth, height=0.275\linewidth, 
        scale only axis,
        xlabel={Qubits per Node ($Q$)},
        ylabel={Number of Nodes ($P$)},
        xmin=1.5, xmax=20.5, 
        ymin=1.5, ymax=19.5, 
        xtick={2,5,8,11,14,17,20},
        ytick={2,4,6,8,10,12,14,16,18},
        enlargelimits=false,
        axis on top,
        view={0}{90},
        point meta min=0, 
        point meta max=190,
    ]
        \nextgroupplot[title={No threshold ($t\rightarrow\infty$)}]       
        \addplot [matrix plot*, mesh/cols=19, mesh/ordering=x varies, point meta=explicit] 
        table [x=qubits_per_node, y=nodes, meta index=2, col sep=comma] {tikz/results/csv/eprs_per_node.csv};      
        \addplot [
            only marks, mark=none,
            nodes near coords,
            nodes near coords align=center,
            point meta=explicit,
            restrict expr to domain={\thisrowno{2}}{-1:99}, 
            nodes near coords style={
                font=\tiny,
                color=black, 
                /pgf/number format/fixed,
                /pgf/number format/precision=0,
            }
        ]
        table [x=qubits_per_node, y=nodes, meta index=2, col sep=comma] {tikz/results/csv/eprs_per_node.csv};
        
        \nextgroupplot[title={Threshold $t=7$}]       
        \addplot [matrix plot*, mesh/cols=19, mesh/ordering=x varies, point meta=explicit] 
        table [x=qubits_per_node, y=nodes, meta index=3, col sep=comma] {tikz/results/csv/eprs_per_node.csv};   
        \addplot [
            only marks, mark=none,
            nodes near coords,
            nodes near coords align=center,
            point meta=explicit,
            nodes near coords style={
                font=\tiny,
                color=black,
                /pgf/number format/fixed,
                /pgf/number format/precision=0,
            }
        ] 
        table [x=qubits_per_node, y=nodes, meta index=3, col sep=comma] {tikz/results/csv/eprs_per_node.csv};
        \nextgroupplot[
            title={Threshold $t=3$},
            colorbar, 
            colorbar style={
                title={EPR Pairs \\ per Node},          
                title style={
                    align=center,       
                    yshift=-2pt,         
                    font=\footnotesize, 
                    rotate=0            
                }
            }
        ]
        \addplot [matrix plot*, mesh/cols=19, mesh/ordering=x varies, point meta=explicit] 
        table [x=qubits_per_node, y=nodes, meta index=4, col sep=comma] {tikz/results/csv/eprs_per_node.csv};
        \addplot [
            only marks, mark=none,
            nodes near coords,
            nodes near coords align=center,
            point meta=explicit,
            nodes near coords style={
                font=\tiny,
                color=black,
                /pgf/number format/fixed,
                /pgf/number format/precision=0,
            }
        ] 
        table [x=qubits_per_node, y=nodes,  meta index=4, col sep=comma] {tikz/results/csv/eprs_per_node.csv};
    \end{groupplot}
\end{tikzpicture}
    \caption{Density of entanglement resources, quantified as the number of EPR pairs generated \textit{per computing node}, as a function of the local register size $Q$ and the total number of nodes $P$. Note that \textbf{higher values correspond to increased communication overhead and resource consumption}. The leftmost panel (Unbounded) shows a linear increase in local resource demand with respect to $P$, implying a quadratic global complexity. The central ($t=7$) and rightmost ($t=3$) panels demonstrate how the pruning threshold effectively saturates this cost, rendering the resource consumption independent of the network size for sufficiently large distributed systems.}
    \label{fig:results:epr_per_node}
\end{figure*}
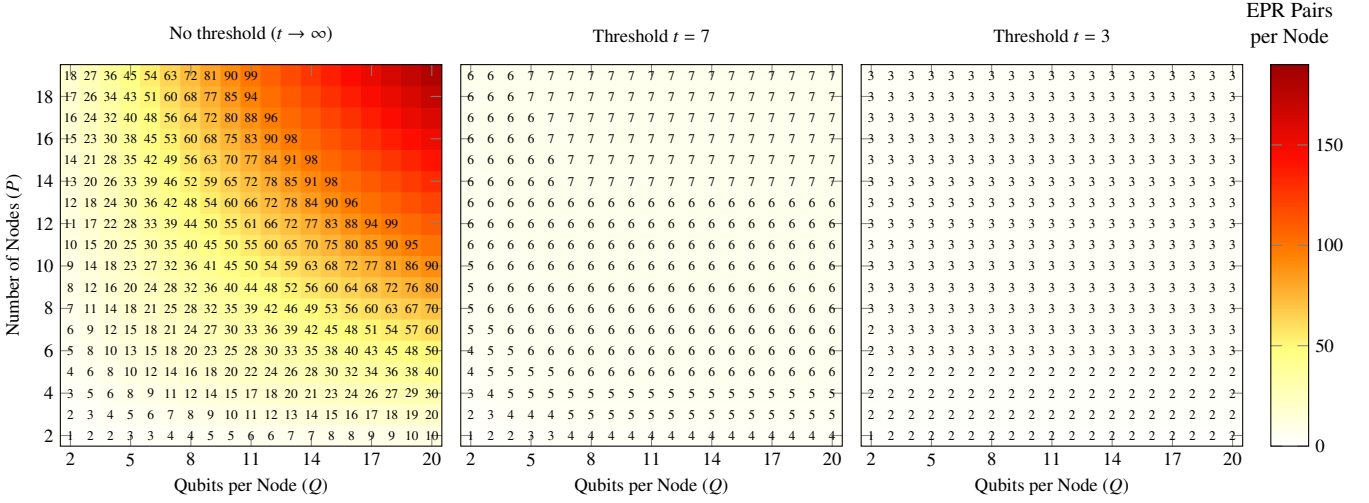

As illustrated in Figure~\ref{fig:coupling_map}, the physical hardware features a sparse and restricted topological coupling map. However, the exact iQFT theoretically demands an all-to-all connectivity graph. When higher thresholds are applied, a larger number of controlled-phase gates are retained in the circuit. Transpiling these extended logical interactions to the restricted physical topology requires the insertion of extensive SWAP-gate routing networks. This routing drastically increases the physical circuit depth, amplifying the accumulation of depolarizing and readout errors, which quickly eclipse the theoretical gains in algorithmic fidelity.

It is important to emphasize that at $t = 16$, the point at which the ideal simulator attains perfect accuracy, the circuit becomes mathematically equivalent to the conventional, state-of-the-art iQFT. The degraded performance observed on noisy backends precisely at this value of $t$ therefore indicates the intrinsic fragility of the standard iQFT in NISQ backends, rather than any fundamental shortcoming of the proposed pruning strategy. Instead, these findings substantiate the essential role of the introduced communication horizon. By intentionally enforcing a more conservative threshold, the architecture drastically reduces the algorithmic demand for non-local interactions. On NISQ backends characterized by sparse coupling maps, this reduction translates directly into a significantly smaller physical circuit depth by avoiding extensive SWAP-gate routing. Consequently, the communication horizon suppresses the otherwise accumulation of physical noise, providing a critical operational trade-off that maximizes the robust extraction of computational information on NISQ hardware.

\subsection{Entanglement Resource Scaling}

Figure~\ref{fig:results:epr_per_node} presents a comparative analysis of the entanglement resources required for the distributed iQFT, expressed as the number of EPR pairs generated \textit{per node}. The results are depicted as heatmaps varying with the local system size ($Q$) and the number of computing nodes ($P$) for the unbounded case (no pruning) and two pruned configurations ($t=7$ and $t=3$).

It is crucial to emphasize that the metric displayed corresponds to the EPR pairs generated \textit{per individual node}. While the total entanglement consumption across the entire network would naturally scale with $P$, the per-node burden provides a more insightful measure of the local communication overhead. As observed in the leftmost panel (unbounded case), without pruning, the resource demand per node grows linearly with both $Q$ and $P$. This leads to a quadratic aggregate complexity $\mathcal{O}(P^2)$, where a system of $P=20$ nodes with $Q=20$ qubits necessitates $\approx180$ EPR pairs per node.

The central and rightmost panels demonstrate the efficacy of the proposed pruning strategy. By applying a filtering threshold $t$, this initially quadratic growth is effectively suppressed. Consistent with the theoretical analysis in Section~\ref{sec:methods:communication_efficiency}, the threshold acts as a cutoff for long-distance interactions. For stricter thresholds (e.g., $t=3$), the number of EPR pairs per node saturates rapidly, becoming independent of the system size $P$. Even a moderate threshold of $t=7$ caps the requirement at a maximum of approximately 7 EPR pairs per node, regardless of the network size. This transition from quadratic to linear global scaling (constant per node) transforms the all-to-all entanglement requirement into a localized resource pattern, shifting the bottleneck from entanglement generation rates to local gate fidelities and making the architecture highly scalable.

\subsection{Coupling Ratio and Node Decoupling}

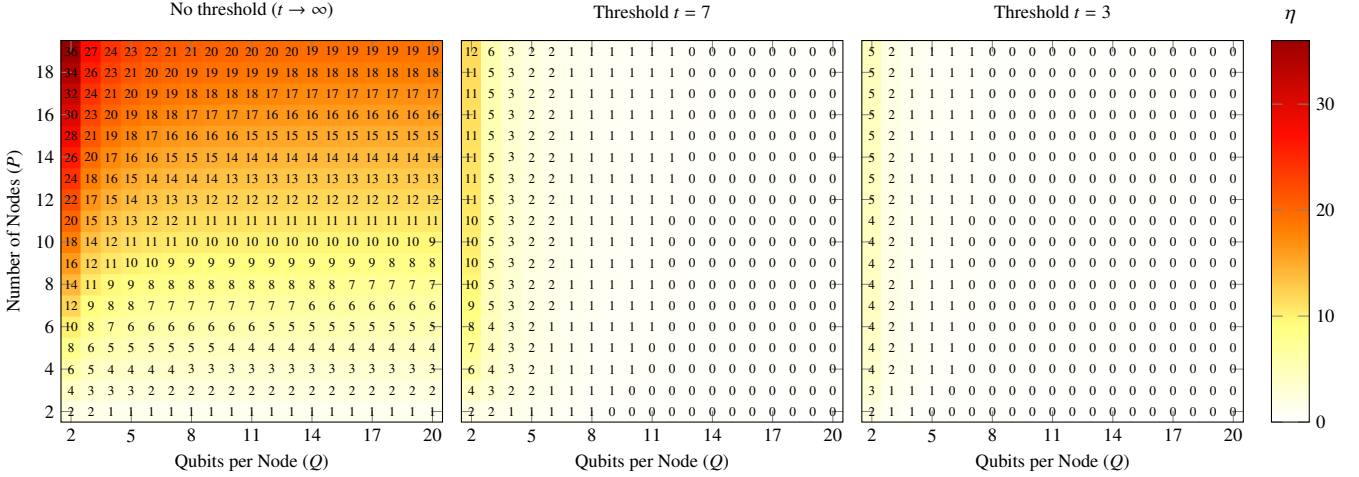
\begin{figure*}[t]
    \centering
    \scriptsize
\begin{tikzpicture}
    \pgfplotsset{
        colormap={fire}{
            color(0cm)=(white);
            color(0.25cm)=(yellow!50!white);
            color(0.5cm)=(orange);
            color(0.75cm)=(red!90!orange);
            color(1cm)=(red!60!black)
        }
    }
    \begin{groupplot}[
        group style={
            group size=3 by 1,  
            horizontal sep=0.25cm, 
            x descriptions at=edge bottom, 
            y descriptions at=edge left,   
        },
        width=0.275\linewidth, height=0.275\linewidth, 
        scale only axis,
        xlabel={Qubits per Node ($Q$)},
        ylabel={Number of Nodes ($P$)},
        xmin=1.5, xmax=20.5, 
        ymin=1.5, ymax=19.5, 
        xtick={2,5,8,11,14,17,20},
        ytick={2,4,6,8,10,12,14,16,18},
        enlargelimits=false,
        axis on top,
        view={0}{90},
        point meta min=0, 
        point meta max=36,
    ]
        \nextgroupplot[title={No threshold ($t\rightarrow\infty$)}]       
        \addplot [matrix plot*, mesh/cols=19, mesh/ordering=x varies, point meta=explicit] 
        table [x=qubits_per_node, y=nodes, meta index=2, col sep=comma] {tikz/results/csv/cz_gates_growth_ratio.csv};      
        \addplot [
            only marks, mark=none,
            nodes near coords,
            nodes near coords align=center,
            point meta=explicit,
            restrict expr to domain={\thisrowno{2}}{-1:99}, 
            nodes near coords style={
                font=\tiny,
                color=black, 
                /pgf/number format/fixed,
                /pgf/number format/precision=0,
            }
        ]
        table [x=qubits_per_node, y=nodes, meta index=2, col sep=comma] {tikz/results/csv/cz_gates_growth_ratio.csv};
        
        \nextgroupplot[title={Threshold $t=7$}]       
        \addplot [matrix plot*, mesh/cols=19, mesh/ordering=x varies, point meta=explicit] 
        table [x=qubits_per_node, y=nodes, meta index=3, col sep=comma] {tikz/results/csv/cz_gates_growth_ratio.csv};   
        \addplot [
            only marks, mark=none,
            nodes near coords,
            nodes near coords align=center,
            point meta=explicit,
            nodes near coords style={
                font=\tiny,
                color=black,
                /pgf/number format/fixed,
                /pgf/number format/precision=0,
            }
        ] 
        table [x=qubits_per_node, y=nodes, meta index=3, col sep=comma] {tikz/results/csv/cz_gates_growth_ratio.csv};
        \nextgroupplot[
            title={Threshold $t=3$},
            colorbar, 
            colorbar style={
                title={$\eta$},          
                title style={
                    align=center,       
                    yshift=-2pt,         
                    font=\footnotesize, 
                    rotate=0            
                }
            }
        ]
        \addplot [matrix plot*, mesh/cols=19, mesh/ordering=x varies, point meta=explicit] 
        table [x=qubits_per_node, y=nodes, meta index=4, col sep=comma] {tikz/results/csv/cz_gates_growth_ratio.csv};
        \addplot [
            only marks, mark=none,
            nodes near coords,
            nodes near coords align=center,
            point meta=explicit,
            nodes near coords style={
                font=\tiny,
                color=black,
                /pgf/number format/fixed,
                /pgf/number format/precision=0,
            }
        ] 
        table [x=qubits_per_node, y=nodes,  meta index=4, col sep=comma] {tikz/results/csv/cz_gates_growth_ratio.csv};
    \end{groupplot}
\end{tikzpicture}
    \caption{Evolution of the Coupling Ratio $\eta$, defined as the ratio between remote and local controlled-phase gates ($\eta = N_{remote}/N_{local}$), as a function of the pruning threshold. Note that \textbf{higher values represent a higher dependency on inter-node communication}, which constitutes the primary bottleneck in distributed architectures. The plot demonstrates how the proposed pruning strategy successfully minimizes this ratio ($\eta \to 0$), effectively decoupling the nodes and shifting the computational load towards efficient local processing.}
    \label{fig:results:cz_gate_ratio}
\end{figure*}

To assess the computational efficiency of the distributed implementation relative to the communication overhead, we analyze the \textit{Coupling ratio} ($\eta$ defined in Eq.~\ref{eq:eta_locality_ratio}). In distributed quantum architectures, the quantum interconnect bandwidth and latency constitute the primary performance bottlenecks, often orders of magnitude slower than local gate operations. An ideal distributed algorithm should strive for $\eta \to 0$, indicating that the computational workload is dominated by local processing.

Figure~\ref{fig:results:cz_gate_ratio} illustrates the evolution of $\eta$ as a function of the pruning threshold $t$. In the standard, unpruned iQFT, the algorithm exhibits a high dependency on remote operations, with $\eta$ values reaching $\approx 35$ for large systems, indicating that inter-node communication dominates the execution time. 

A central contribution of this work is the establishment of the communication horizon, $d_\text{max}(Q,\varepsilon)$ (Eq.~\ref{eq:communication_horizon}), which acts as a decisive control parameter. Since the amplitude of the controlled-phase rotation $\theta_k = \pi/2^{k}$ decays exponentially with the logical distance $k$, the pruning threshold $t$ effectively imposes a hard limit on the interaction range. This mechanism preferentially filters out long-range remote operations typically corresponding to high $k$. The application of this optimization drastically suppresses the ratio. For thresholds $t \le 7$, $\eta$ approaches asymptotic zero ($\eta \to 0$). This reduction signifies a fundamental decoupling of the nodes: the computational workload is almost entirely confined to the local QPU, with inter-node links serving only for sparse, nearest-neighbor synchronization. This decoupling ensures that the latency of the quantum channel does not exponentially compound.

\subsection{Accuracy under Connectivity Constraints}

Figure~\ref{fig:results:dmax} depicts the impact of imposing a fixed communication horizon $d_\text{max}$ on the algorithmic accuracy. The heatmaps quantify the theoretical error bound defined in Eq.~\ref{eq:error} as a function of the local node capacity ($Q$) and the maximum allowed communication distance ($d_\text{max}$ defined in Eq.~\ref{eq:communication_horizon}).

The concept of the communication horizon emerges as a decisive parameter for distributed architecture design. As evidenced by the rapid error convergence in the heatmap, the algorithmic error decays precipitously as the communication horizon increases. Even for moderate horizon values, the fidelity of the distributed iQFT remains extremely high. For instance, in a system with local capacity $Q=9$ qubits per node, restricting the interaction range to just $d_\text{max}=2$ neighboring nodes yields an error bound of approximately $2^{-10}$. Extending connectivity beyond a small neighborhood provides diminishing returns in accuracy while incurring significant communication costs. 

A fundamental property revealed by this analysis is the independence of the error bound from the total number of nodes $P$. The accuracy is determined solely by the local register size $Q$ and the communication horizon $d_\text{max}$. This implies that for any target precision $\epsilon$, the connectivity graph does not need to scale with the network size. A constant connectivity radius guarantees a specific error tolerance whether the network consists of $P=20$ or $P=100$ nodes. This independence is a prerequisite for scalable distributed quantum computing, as it prevents the communication overhead from diverging as the system grows.

\begin{figure}
    \centering
    \scriptsize
\begin{tikzpicture}
    \begin{axis}[
        scale only axis,
        width=0.68\linewidth, height=0.68\linewidth, 
        title={Heatmap of $\varepsilon=2^{-i}$},
        ylabel={Qubits per node ($Q$)},
        xlabel={Communication horizon ($d_\text{max}$)},
        xmin=1.5, xmax=14.5, 
        ymin=1.5, ymax=14.5,
        xtick={2,3,...,14}, 
        ytick={2,3,...,14},
        colorbar,
        colorbar style={
            ylabel={$t$},
            ytick={1, 20, 40, ..., 180},
            ylabel style={
                rotate=270,
                at={(axis description cs:0.55, 1.1)}, 
                anchor=north
            }
        },
        colormap={fire}{
            color(0cm)=(red!60!black);                 
            color(0.25cm)=(red!90!orange);    
            color(0.5cm)=(orange);              
            color(0.75cm)=(yellow!50!white);      
            color(1cm)=(white)           
        },
        point meta max=183, 
        point meta min=1,
        view={0}{90}, 
        enlargelimits=false,
        axis on top, 
    ]
        \addplot [
            matrix plot*, 
            mesh/rows=14, 
            mesh/ordering=y varies, 
            point meta=explicit, 
        ] table [
            x=q, 
            y=N_prime, 
            meta=log2_epsilon, 
            col sep=comma
        ] {tikz/results/csv/epsilon_heatmap_log.csv};
        \addplot [
            only marks,
            mark=none,
            nodes near coords,
            nodes near coords align=center,
            visualization depends on={value \thisrow{log2_epsilon} \as \myval},
            nodes near coords style={
                /utils/exec={
                    \pgfmathparse{\myval < 100 ? "white" : "black"} 
                    \edef\mycolor{\pgfmathresult} 
                    \pgfmathparse{\myval < 100 ? 1 : 0}
                    \ifnum\pgfmathresult=1
                        \pgfkeysalso{font=\scriptsize\bfseries, text=\mycolor}
                    \else
                        \pgfkeysalso{font=\tiny\bfseries, text=\mycolor}
                    \fi
                },
                /pgf/number format/fixed,
                /pgf/number format/precision=0
            },
            point meta=explicit
        ] table [
            x=N_prime, 
            y=q, 
            meta=log2_epsilon, 
            col sep=comma
        ] {tikz/results/csv/epsilon_heatmap_log.csv};
    \end{axis}
\end{tikzpicture}
    \caption{Impact of the communication horizon $d_\text{max}$ and local register size $Q$ on the algorithmic accuracy. The values within the cells represent the \textbf{negated exponent} $i$ of the theoretical error bound (such that $\varepsilon \approx 2^{-i}$). Therefore, \textbf{higher values indicate lower errors and higher fidelity}. A rapid convergence is observed, showing that small horizons are sufficient to achieve high precision.}
    \label{fig:results:dmax}
\end{figure}
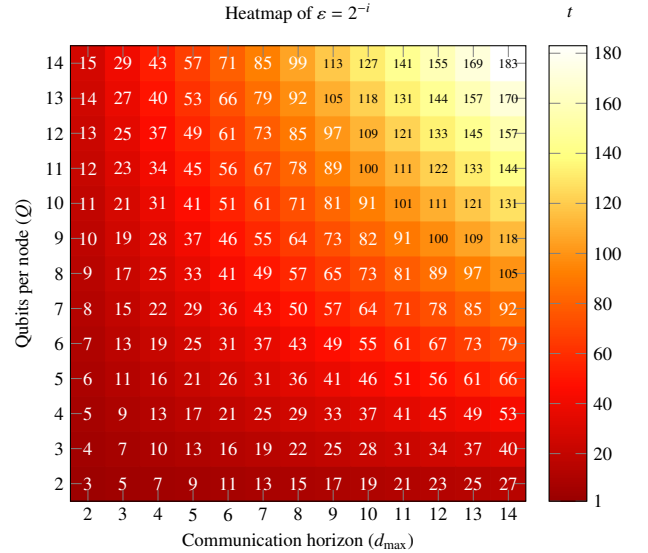

\subsection{Communication Overhead}

\begin{figure*}
    \centering
    \scriptsize
\begin{tikzpicture}
    \pgfplotsset{
        colormap={fire}{
            color(0cm)=(white);
            color(0.25cm)=(yellow!50!white);
            color(0.5cm)=(orange);
            color(0.75cm)=(red!90!orange);
            color(1cm)=(red!60!black)
        }
    }
    \begin{groupplot}[
        group style={
            group size=3 by 1,  
            horizontal sep=0.25cm, 
            x descriptions at=edge bottom, 
            y descriptions at=edge left,   
        },
        width=0.275\linewidth, height=0.275\linewidth, 
        scale only axis,
        xlabel={Qubits per Node ($Q$)},
        ylabel={Number of Nodes ($P$)},
        xmin=1.5, xmax=20.5, 
        ymin=1.5, ymax=19.5, 
        xtick={2,5,8,11,14,17,20},
        ytick={2,4,6,8,10,12,14,16,18},
        enlargelimits=false,
        axis on top,
        view={0}{90},
        point meta min=0, 
        point meta max=4,
    ]
        \nextgroupplot[title={No threshold ($t\rightarrow\infty$)}]       
        \addplot [matrix plot*, mesh/cols=19, mesh/ordering=x varies, point meta=explicit] 
        table [x=qubits_per_node, y=nodes, meta index=2, col sep=comma] {tikz/results/csv/communication_overhead.csv};      
        \addplot [
            only marks, mark=none,
            nodes near coords,
            nodes near coords align=center,
            point meta=explicit,
            restrict expr to domain={\thisrowno{2}}{-1:99}, 
            nodes near coords style={
                font=\tiny,
                color=black, 
                /pgf/number format/fixed,
                /pgf/number format/precision=0,
            }
        ]
        table [x=qubits_per_node, y=nodes, meta index=2, col sep=comma] {tikz/results/csv/communication_overhead.csv};
        
        \nextgroupplot[title={Threshold $t=7$}]       
        \addplot [matrix plot*, mesh/cols=19, mesh/ordering=x varies, point meta=explicit] 
        table [x=qubits_per_node, y=nodes, meta index=3, col sep=comma] {tikz/results/csv/communication_overhead.csv};   
        \addplot [
            only marks, mark=none,
            nodes near coords,
            nodes near coords align=center,
            point meta=explicit,
            nodes near coords style={
                font=\tiny,
                color=black,
                /pgf/number format/fixed,
                /pgf/number format/precision=0,
            }
        ] 
        table [x=qubits_per_node, y=nodes, meta index=3, col sep=comma] {tikz/results/csv/communication_overhead.csv};
        \nextgroupplot[
            title={Threshold $t=3$},
            colorbar, 
            colorbar style={
                title={$\gamma$},          
                title style={
                    align=center,       
                    yshift=-2pt,         
                    font=\footnotesize, 
                    rotate=0            
                }
            }
        ]
        \addplot [matrix plot*, mesh/cols=19, mesh/ordering=x varies, point meta=explicit] 
        table [x=qubits_per_node, y=nodes, meta index=4, col sep=comma] {tikz/results/csv/communication_overhead.csv};
        \addplot [
            only marks, mark=none,
            nodes near coords,
            nodes near coords align=center,
            point meta=explicit,
            nodes near coords style={
                font=\tiny,
                color=black,
                /pgf/number format/fixed,
                /pgf/number format/precision=0,
            }
        ] 
        table [x=qubits_per_node, y=nodes,  meta index=4, col sep=comma] {tikz/results/csv/communication_overhead.csv};
    \end{groupplot}
\end{tikzpicture}
    \caption{Global Communication Overhead $\gamma$, defined as the ratio between the number of auxiliary communication operations (entanglement generation, swapping, and corrections) and the logical computational gates ($N_{\text{comm}}/N_\text{comp}$). This metric indicates the multiplicative factor by which the communication burden exceeds the useful processing work. Note that \textbf{higher values are detrimental}, representing a regime where the system execution is dominated by network maintenance rather than algorithmic progress.}   
    \label{fig:results:overhead}
\end{figure*}
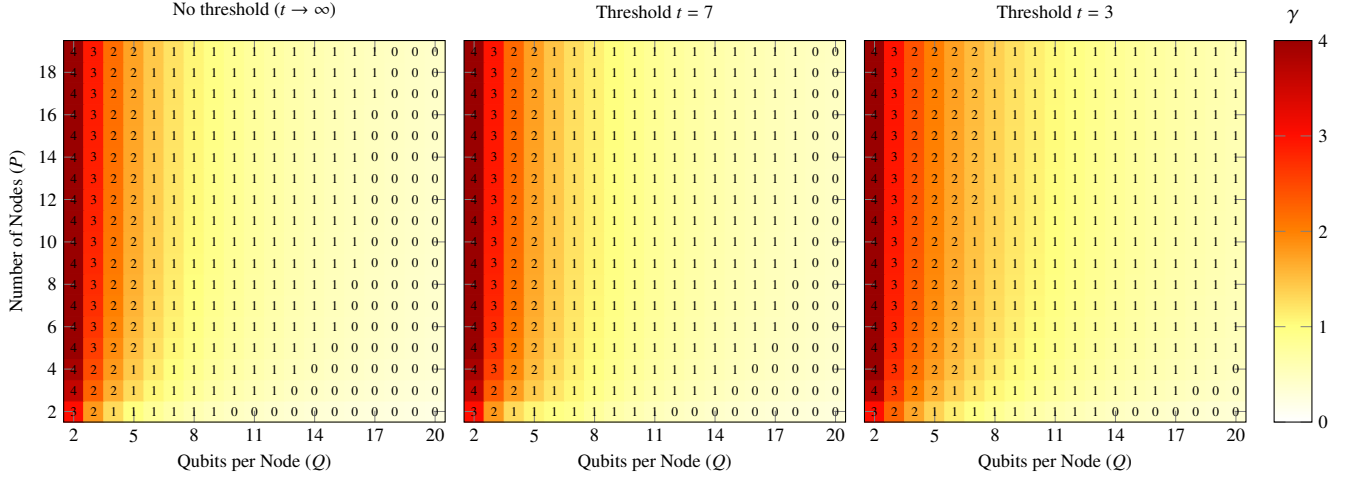

Figure~\ref{fig:results:overhead} analyzes the Communication Overhead ratio ($\gamma$) across different system sizes and pruning configurations. A crucial observation from these results is that the overhead factor remains relatively invariant with respect to the pruning threshold $t$.

This stability implies that the application of the threshold does not alter the fundamental architectural efficiency of the distributed protocol. The rationale lies in the fixed operational cost of the \textit{Telegate} mechanism: each executed remote Controlled-Phase gate necessitates a specific, proportional set of auxiliary communication operations (Bell pair generation, entanglement swapping, and corrections). Therefore, when the pruning strategy removes a logical gate, it simultaneously eliminates the associated communication burden.

Consequently, while the \textit{total count} of operations decreases drastically, the \textit{ratio} between communication and computation tasks ($\gamma$) remains stable. This confirms that the overhead is an intrinsic property of the chosen quantum interconnect protocol and is not artificially inflated or distorted by the pruning approximation.

\subsection{Comparison with State-of-the-Art and Theoretical Bounds}

To contextualize the impact of the proposed architecture, Table~\ref{tab:literature_comparison} summarizes the communication complexity and core strategies of existing distributed QFT/iQFT frameworks.

\begin{table}[h]
    \centering
    \renewcommand{\arraystretch}{1.5}
\footnotesize
\begin{tabular}{@{}lclc@{}}
\toprule
\textbf{Reference}                                                     & \textbf{Paradigm} & \textbf{\begin{tabular}[c]{@{}l@{}}Key\\ Concept\end{tabular}}                    & \textbf{\begin{tabular}[c]{@{}c@{}}Comm. \\ Scaling\end{tabular}} \\ \addlinespace \midrule
\begin{tabular}[c]{@{}l@{}}Standard Distributed \\ iQFT\end{tabular}              & Exact             & \begin{tabular}[c]{@{}l@{}}Baseline all-to-all\\ mapping\end{tabular}             & $\mathcal{O}(P^2)$                                                \\ \addlinespace
\begin{tabular}[c]{@{}l@{}}Compiler Batching \\ \cite{Kaur2025OptimizedQuantumCircuit}\end{tabular}      & Exact             & \begin{tabular}[c]{@{}l@{}}Remote-gate batching\\ and scheduling\end{tabular}     & $\mathcal{O}(P^2)$                                                \\ \addlinespace
\begin{tabular}[c]{@{}l@{}}Theoretical Bound \\ \cite{Ebnenasir2025LowerBounds}\end{tabular} & Exact             & \begin{tabular}[c]{@{}l@{}}Lower bound proof \\ for clique networks\end{tabular}  & $\Omega(P^2)$                                                     \\ \addlinespace
Proposed work                                                                     & Approximate       & \begin{tabular}[c]{@{}l@{}}Communication \\ horizon ($d_\text{max}$)\end{tabular} & $\mathcal{O}(P)$                                         \\ \bottomrule
\end{tabular}%
    \caption{Comparison of distributed QFT/iQFT approaches and communication scaling across $P$ nodes.}
    \label{tab:literature_comparison}
\end{table}

An analysis of the distributed quantum computing literature reveals that existing solutions operate under distinct methodological paradigms. In this context, the work of Ebnenasir \cite{Ebnenasir2025LowerBounds} constitutes a fundamental theoretical predecessor that establishes the strict lower bounds for the communication costs of distributing the exact QFT on clique networks. This study rigorously proves that any exact implementation is mathematically constrained to an aggregate communication complexity of $\Omega(P^2)$ due to the dense, all-to-all logical connectivity of the circuit. The linear $\mathcal{O}(P)$ communication scaling achieved by our proposed communication-efficient iQFT is fully consistent with Ebnenasir's theoretical limit and does not violate it, since his theorem applies strictly to the exact formulation. By strategically transitioning to an approximate transform bounded by a fractional phase tolerance $\varepsilon$, the introduction of the communication horizon ($d_\text{max}$) deliberately circumvents this mathematical barrier. This demonstrates that long-range interactions are structurally dispensable, successfully localizing the network overhead.

Conversely, this algorithmic reduction is fundamentally complementary to the compiler-level optimization frameworks proposed by Kaur et al. \cite{Kaur2025OptimizedQuantumCircuit}. The approach by Kaur et al. focuses on mitigating network latency by batching and parallelizing the mandatory inter-node interactions required by the circuit. In contrast, our pruning strategy intervenes at an earlier stage of the software stack, drastically reducing the absolute volume of non-local operations before the circuit is ever processed by the scheduler. Consequently, both methodologies are not mutually exclusive but highly synergetic: an optimized compiler can seamlessly apply remote-gate batching techniques over the circuit already simplified by our communication horizon, further maximizing the global efficiency of the distributed architecture.
\section{Conclusions}
\label{sec:conclusions}

In this work, we have presented a communication-aware optimization of the Inverse Quantum Fourier Transform (iQFT) designed specifically for distributed quantum architectures. By rigorously analyzing the geometric decay of the controlled-phase rotation angles, we introduced a pruning strategy that transforms the algorithm's connectivity requirements from a global all-to-all topology to a bounded nearest-neighbor interaction model.

Our theoretical and empirical analyses demonstrate that establishing a communication horizon, $d_\text{max}$, allows for a drastic reduction in the consumption of entanglement resources. Specifically, we showed that the number of EPR pairs per node transitions from a linear scaling $\mathcal{O}(P)$ with respect to the network size to a constant saturation level $\mathcal{O}(1)$, effectively making the resource overhead independent of the total number of nodes in the system. Furthermore, the analysis of the Coupling Ratio ($\eta$) confirms that our method successfully decouples the processing nodes, reducing the dependency on high-latency inter-node links to near-zero values while maintaining algorithmic fidelity well within the theoretical error bounds derived.

Ultimately, this approach provides a scalable pathway for implementing large-scale Quantum Phase Estimation and other Fourier-based primitives on near-term networked quantum processors, offering a tunable trade-off between execution fidelity and communication bandwidth.

\subsection{Future Work}
While the current study utilized state-vector simulation to validate the theoretical error models and resource scaling exactitude, several key areas for future investigation. First, it would be valuable to extend the validation of this algorithm to specialized distributed quantum computing (DQC) simulators such as NetQMPI, CUNQA, or CUDA-Quantum. These platforms would allow for a more granular analysis of network-specific effects, including link latency, generation rates, and the impact of classical control delays on decoherence.

Additionally, future work should address the integration of this optimized iQFT schedule with quantum error correction codes suitable for distributed environments, such as surface codes over noisy interconnects, to further suppress the effective error rates in the presence of realistic channel noise.

{\small
\section*{Declaration of generative AI and AI-assisted technologies in the writing process}

During the preparation of this work the authors used ChatGPT and Gemini in order to improve language and readability. After using this tool, the authors reviewed and edited the content as needed and take full responsibility for the content of the publication.
}

{\small
\section*{Acknowledgments} 
This work was funded by the European Union EuroHPC program with grant agreement 101194491 and by MICIU/AEI/10.13039/501100011033 with project number PCI2025-163229, the Agencia Estatal de Investigación (Spain) (PID2022-141623NB-I00), the Consellería de Cultura, Educación e Ordenación Universitaria Galician Research Center accreditation 2024–2027 ED431G-2023/04, and the European Regional Development Fund (ERDF).
}

\bibliographystyle{elsarticle-num}
\UseRawInputEncoding
\bibliography{references}

\end{document}